\documentclass[journal]{IEEEtran}
\usepackage{graphicx}          
\usepackage{hyperref}
\makeatletter \let\cl@part\relax \makeatother 
\usepackage{amsmath} 
\usepackage{amsthm}
\usepackage{amssymb}
\usepackage{balance}
\usepackage{mathtools}
\usepackage{siunitx}
\usepackage{caption}
\usepackage[usenames,dvipsnames]{xcolor}
\usepackage{pgfplots}
\usepackage{pgfplotstable}
\usepgfplotslibrary{fillbetween,groupplots}
\usepackage{tabularx}
\usepackage{booktabs}
\pgfplotsset{compat = newest}
\usetikzlibrary{arrows,positioning,shapes,intersections,patterns,calc,fit,decorations.pathreplacing,external} 
\usepackage{algpseudocode} 
\usepackage{algorithm}
\newtheorem{propy}{Property}

\newtheorem{rem}{Remark}

\newtheorem{prop}{Proposition}
\newtheorem{cor}{Corollary}
\newtheorem{thm}{Theorem}
\newtheorem{defn}{Definition}

\theoremstyle{definition}
\newtheorem{exam}{Example}

\newcommand\tran{\mkern-2mu\raise1.25ex\hbox{$\scriptscriptstyle\top\hspace{0.5mm}$}\mkern-3.5mu}
\newcommand{\R}{\mathbb{R}}
\newcommand{\N}{\mathbb{N}}
\newcommand{\ND}{\mathcal{N}}
\newcommand{\D}{\mathcal{D}}
\newcommand{\X}{\mathcal{X}}

\newcommand{\bm}[1]{{\boldsymbol{#1}}}
\newcommand{\Verts}[1]{{\left\Vert #1 \right\Vert}}

\DeclareMathOperator{\diag}{diag}
\DeclareMathOperator{\var}{var}
\DeclareMathOperator{\mean}{\mu}
\DeclareMathOperator{\Var}{\Sigma}
\DeclareMathOperator{\Mean}{\bm\mu}
\DeclareMathOperator{\Prob}{P}
\DeclareMathOperator{\ev}{E}
\DeclareMathOperator{\mspe}{MSPE}

\DeclareMathOperator{\geig}{\bar{\lambda}}

\DeclareMathOperator{\tr}{tr}
\makeatletter
\DeclareRobustCommand{\rvdots}{%
  \vbox{
    \baselineskip4\p@\lineskiplimit\z@
    \kern-\p@
    \hbox{.}\hbox{.}\hbox{.}
  }}
  \makeatother
\newcommand{\GP}{\mathcal{GP}}

\newcommand{\vu}{\bm u}
\newcommand{\inputu}{\bm u}
\newcommand{\uk}{{\bm u}_t}
\newcommand{\x}{\bm x}
\newcommand{\xk}{\bm{x}_t}
\newcommand{\xkp}{\bm{x}_{t+1}}
\newcommand{\y}{\bm y}
\newcommand{\yk}{\bm{y}_t}
\newcommand{\ykp}{\bm{y}_{t+1}}
\newcommand{\z}{\bm\zeta}
\newcommand{\zk}{\bm{\zeta}_t}

\newcommand{\f}{\bm{f}}
\newcommand{\dyn}{\bm{f}}
\newcommand{\m}{\bm{m}}
\newcommand{\h}{\bm{h}}
\newcommand{\vxi}{\bm{\xi}}

\newcommand{\ubar}[1]{\underline{#1}}
\usepackage[noabbrev]{cleveref} 
\crefname{propy}{Property}{Properties}
\crefname{exam}{Example}{Examples}
\crefname{rem}{Remark}{Remarks}
\crefname{assum}{Assumption}{Assumptions}
\crefname{prop}{Proposition}{Propositions}
\crefname{cor}{Corollary}{Corollaries}
\crefname{lem}{Lemma}{Lemmas}
\crefname{thm}{Theorem}{Theorems}
\crefname{defn}{Definition}{Definitions}
\crefname{figure}{Fig.}{Fig.}
\Crefname{figure}{Figure}{Figures}
\crefname{equation}{}{}
\Crefname{equation}{Equation}{Equations}

\pgfplotsset{
    xticks from table/.code 2 args={%
        \pgfplotstablegetrowsof{#1}
        \pgfmathsetmacro\pgfplotstabletotalrows{\pgfplotsretval-1}
        \edef\ticklist{}
            \pgfplotstableforeachcolumnelement{#2}\of#1\as\cell{%
                \edef\ticklist{\ticklist\cell\ifnum\pgfplotstablerow<\pgfplotstabletotalrows,\fi}
            }
        \pgfplotsset{xtick=\ticklist}
    },
}
\begin{document}
%
\title{Prediction with Approximated Gaussian Process\\Dynamical Models}

\author{Thomas Beckers and Sandra Hirche
\thanks{T. Beckers is with Department of Electrical and Systems Engineering, University of Pennsylvania, Philadelphia, USA, {\tt\small tbeckers@seas.upenn.edu}. S. Hirche is with the Chair of Information-oriented Control (ITR), Department of Electrical and Computer Engineering, Technical University of Munich, 80333 Munich, Germany, {\tt\small hirche@tum.de}}}

%

\maketitle

\begin{abstract}
The modeling and simulation of dynamical systems is a necessary step for many control approaches. Using classical, parameter-based techniques for modeling of modern systems, e.g., soft robotics or human-robot interaction, is often challenging or even infeasible due to the complexity of the system dynamics. In contrast, data-driven approaches need only a minimum of prior knowledge and scale with the complexity of the system. In particular, Gaussian process dynamical models (GPDMs) provide very promising results for the modeling of complex dynamics. However, the control properties of these GP models are just sparsely researched, which leads to a "blackbox" treatment in modeling and control scenarios. In addition, the sampling of GPDMs for prediction purpose respecting their non-parametric nature results in non-Markovian dynamics making the theoretical analysis challenging. In this article, we present approximated GPDMs which are Markov and analyze their control theoretical properties. Among others, the approximated error is analyzed and conditions for boundedness of the trajectories are provided. The outcomes are illustrated with numerical examples that show the power of the approximated models while the the computational time is significantly reduced.
\end{abstract}


%
\IEEEpeerreviewmaketitle

\section{Introduction}
\IEEEPARstart{M}{odeling} of dynamical systems plays a very import role in the area of control theory. The goal is the derivation of a mathematical model which is based on generated input data and the corresponding output data of the plant. The model is necessary for any model-based control design, such as model predictive control. Besides, a model is required for simulations to evaluate the quality of the control designs and to improve the understanding of the system. To achieve a dynamical model, the output of the model is feedbacked to the model itself. A special class of dynamical models is given by \emph{simulation models}, which do not rely on any data from the plant during the prediciton~\cite{nelles2013nonlinear}. Therefore, these models are suitable to perform predictions independent of the plant for not only simulations but also in control scenarios such as model predictive control. Classical system identification deals with parametric models. If the system contains nonlinearities, there exist various identification techniques, which mostly depends on the structure of the nonlinear elements. For these approaches, a suitable model structure must be selected a priori to achieve useful results. However, there exists a large class of systems which can not be accurately described by parametric models. Especially, for complex systems such as human motion dynamics~\cite{wang2008gaussian,sigal2012loose}, prediction of climate effects~\cite{petelin2013evolving,hassouneh2012non} or structural dynamics~\cite{soize2005comprehensive,lydia2014comprehensive}, non-parametric techniques appear to be more promising.\\
Within the past two decades, Gaussian processes (GPs) have been developed as powerful function regressors. A GP connects every point of a continuous input space with a normally distributed random variable. Any finite group of those infinitely many random variables follows a multivariate Gaussian distribution. The result is a powerful tool for nonlinear function regression without the need of much prior knowledge~\cite{rasmussen2006gaussian}. In contrast to most of the other techniques, GP modeling provides not only a mean function but also a measure for the uncertainty of the prediction. The output is a Gaussian distributed variable which is fully described by the mean and the variance. There are several possibilities to use a Gaussian process for dynamic system modeling. A frequent approach is the state space model which is in general a very efficient model structure. Gaussian process dynamical models (GPDMs) have recently also become a versatile tool in system identification because of their beneficial properties such as the bias variance trade-off and the strong connection to Bayesian mathematics, see~\cite{frigola2014variational}. The Gaussian process state space model (GP-SSM) uses GPs for modeling dynamical systems with state space models, see~\cite{kocijan2005dynamic}, where each state is described by an own GP. The function between the states and the system's outputs are modeled by another GP or a parametric structure. Alternatives are given by nonlinear identification models such as NFIR~\cite{ackermann2011nonlinear}, NARX~\cite{kocijan2005dynamic} or nonlinear output error (NOE) models~\cite{kocijan2011output}. In comparison to the other models, the NOE has the advantage of being a simulation model such as the GP-SSM.
Although the application of Gaussian process dynamical models increases in control theory, e.g., for adaptive control and model predictive control~\cite{rogers2011adaptive,kocijan2004gaussian,hewing2018cautious}, the theoretical properties of these GPDMs are only sparsely researched. However, the theoretical properties are crucial for further investigations in robustness, stability and performance of control approaches based on GPDMs~\cite{kocijan2016modelling,avzman2008non}.\\
In many works where GPs are considered as dynamical model, only the mean function of the process is employed, for instance in~\cite{wang2005gaussian} and~\cite{chowdhary2013bayesian}. This is mainly because a GPDM is often used for replacing a deterministic model in already existing model-based control approaches. In~\cite{beckers:ecc2016} some basic theoretical properties for deterministic GP-SSMs are derived. However, GPDMs contain a much richer description of the underlying dynamics but also the uncertainty about the model itself when the full probabilistic representation is considered. In~\cite{medina2015synthesizing,beckers2019automatica} control laws are derived which explicitly take the uncertainty of GPDMs into account but without investigation of the control properties of the models. In order to ensure the applicability of GPDMs, classical control theory properties are required, see~\cite{kocijan2016modelling} and~\cite{avzman2008non}. Such basic properties of a dynamical system are, among others, the existence of boundedness conditions. In~\cite{beckers:cdc2016} some basic boundedness properties for simplified probabilistic GP-SSMs are presented. However, it turns out that the analysis of GPDMs is challenging as its usage in simulations requires the sampling of an infinite-dimensional object which is not possible without further simplifications, e.g., discrete sampling with interpolation. To overcome this issue, the authors of~\cite{frigola2013bayesian} propose to marginalize out the transition functions to respect the nonparametric nature of the model using Particle Markov Chain Monte Carlo (PMCMC). However, the resulting non-Markovian dynamics is strongly undesired in control systems as it leads to theoretical and practical issues: the analyse tools in control are mostly suited for Markovian systems such as the Lyapunov stability and the dependencies across time results in computational time and memory issues for long-time simulations. Even though effective sampling of GPDMs has recently gained attention in the machine learning community, e.g.~\cite{wilson2020efficiently,cheng2017variational}, the control related implications are still open.\\
\textbf{Contributions:} The contribution of this article is the introduction of approximated GP-SSM and GP-NOE models to recover the Markovian property. For this purpose, the set of past states for the prediction of the next state ahead is shortened to a finite subset. We show that control relevant properties such as boundedness for the open- and closed-loop are preserved in the transition from the true model to the approximated, Markovian model. In addition, it is guaranteed that the predicted uncertainty of the approximated model overestimated the true uncertainty. Furthermore, upper bounds for the approximation error expressed by the means square prediction error and the Kullback–Leibler divergence are presented. In two case studies, we discuss the behavior of the approximated models and highlight their benefits.\\
\textbf{Notation:} Vectors and vector-valued functions are denoted with bold characters~$\bm{v}$. The notation~$[\bm{a};\bm{b}]$ is used for~$[\bm{a}^\top,\bm{b}^\top]^\top$ and~$\x_{1:n}$ denotes~$[\x_1,\ldots,\x_n]$. Capital letters~$A$ describes matrices. The matrix~$I$ is the identity matrix in appropriate dimension. The expression~$\mathcal{N}(\mu,\Sigma)$ describes a normal distribution with mean~$\mu$ and covariance~$\Sigma$. $N_+$ denotes the positive natural numbers.\\
The remainder of the article is structured as follows. In~\cref{sec:back}, the background of GPDMs as well as their sampling procedure are introduced.~\Cref{sec:approx} presents the novel approximated GPDMs and the approximation error. The boundedness of the presented models is analyzed in~\cref{sec:bound}, followed by a discussion. Finally, two case studies demonstrate the applicability.
\section{Preliminaries and Definitions}
\label{sec:back}
In this article, we focus on Gaussian process based dynamic models. Thus, we start with a brief introduction to GPs as they are the central part of the model.
\subsection{Gaussian Process Models}
Let~$(\Omega, \mathcal{F},P)$ be a probability space with the sample space~$\Omega=\R^n,n\in\N$, the corresponding~$\sigma$-algebra~$\mathcal{F}$ and the probability measure~$P$. Consider a vector-valued, unknown function~$\y=\f(\bm{z})$ with~$\f\colon \R^n\to \R^{n_f}$ and~$\y\in\R^{n_f}$. The measurement~$\tilde{\y}\in\R^{n_f}$ of the function is corrupted by Gaussian noise~$\bm\eta\in\R^{n_f}$, i.e.,
\begin{align}
	\tilde{\y}&=\f(\bm{z})+\bm\eta,\quad\bm\eta\sim\mathcal{N}(\bm 0,\Sigma_n)
\end{align}
with the positive definite matrix~$\Sigma_n=\diag (\sigma_{1}^2,\ldots,\sigma_{n_f}^2)$. To generate the training data, the function is evaluated at~${n_\D}$ input values~$\{\bm{z}^{\{j\}}\}_{j=1}^{n_\D}$. Together with the resulting measurements~$\{\tilde{\y}^{\{j\}}\}_{j=1}^{n_\D}$, the whole training data set is described by~$\D=\{X,Y\}$ with the input training matrix~$X=[\bm{z}^{\{1\}},\bm{z}^{\{2\}},\ldots,\bm{z}^{\{{n_\D}\}}]\in\R^{n\times {n_\D}}$ and the output training matrix~$Y=[\tilde{\y}^{\{1\}},\tilde{\y}^{\{2\}},\ldots,\tilde{\y}^{\{{n_\D}\}}]^\top\in\R^{{n_\D}\times {n_f}}$. Now, the objective is to predict the output of the function~$\f(\bm{z}^*)$ at a test input~$\bm{z}^*\in\R^n$. The underlying assumption of GP modeling is, that the data can be represented as a sample of a multivariate Gaussian distribution using a kernel function $k$. The joint distribution of the~$i$-th component of~$\f(\bm{z}^*)$ is\footnote{For notational convenience, we simplify~$K(X,X)$ to~$K$}
\begin{align}
	\begin{bmatrix}\! Y_{:,i} \\ {f}_i(\bm{z}^*) \!\end{bmatrix}\!\sim\! \mathcal{N}\!\left(\!\begin{bmatrix}\m(X) \\ {m}(\bm{z}^*)\end{bmatrix}\!\!,\!\! \begin{bmatrix} K(X,X)+\sigma_i^2 I & \bm{k}(\bm{z}^*\!,X)\\ \bm{k}(\bm{z}^*\!,X)\tran & k(\bm{z}^*\!,\bm{z}^*) \end{bmatrix}\!\right)\label{sec2:for:joint_dist}
\end{align} 
with the kernel~$k\colon\R^n\times\R^n\to\R$ as a measure of the correlation of two points~$(\bm x,\bm x^\prime)$. The mean function is given by a continuous function~$m\colon\R^{n}\to\R$ and the vector of mean functions~$\bm{m}\colon\R^{n\times {n_\D}}\to\R^{n_\D}$ by~$\bm{m}(X)=[m(X_{:,1});\ldots;m(X_{:,{n_\D}})]$. The kernel function is the central part of the kernel trick, which transforms the data to a higher dimensional feature space~$\Upsilon$, see~\cref{fig:kernel}, without knowing the actual transformation~$\phi\colon R^{n_\D}\to\Upsilon$ since~$k(x,x')=\left\langle \phi(x),\phi(x')\right\rangle~$. Then, a linear regression is performed in the feature space and the output is transformed back.
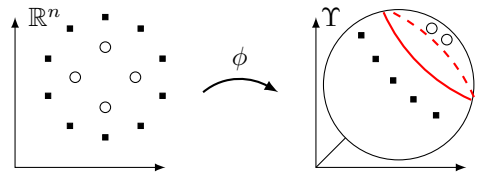
\begin{figure}[ht]
	\begin{center}
		\tikzsetnextfilename{kernel_trick}
		\vspace{0.15cm}
		\begin{tikzpicture}[scale=1,>=latex]
	\def\so{4}
	\def\rwi{2}
	\def\rhe{1.3}
	\def\rx{1.5}
	\def\ry{1.6}
	\def\centerarc[#1](#2)(#3:#4:#5)
    { \draw[#1] ($(#2)+({#5*cos(#3)},{#5*sin(#3)})$) arc (#3:#4:#5); }
	\coordinate (y) at (0,2);
    \coordinate (x) at (2,0);
    \draw[<->] (y) node[above] {} -- (0,0) --  (x) node[right] {};
    \foreach \i in {0,...,3}
{
    \draw ({0.4*sin(\i/2*pi r)+1.2},{0.4*cos(\i/2*pi r)+1.2}) circle (2pt); 
}
    \foreach \i in {1,...,10}
{
	\node [draw,rectangle,minimum width=2pt,minimum height=2pt,inner sep=0pt,fill](c\i) at ({0.8*sin(\i/5*pi r)+1.2},{0.8*cos(\i/5*pi r)+1.2}){};
}
	\coordinate (y1) at (\so,2);
    \coordinate (x1) at (2+\so,0);
    \coordinate (z1) at (1.5+\so,1.5);
    \draw[<->] (y1) node[] {} -- (\so,0) --  (x1) node[] {};
    \draw[->]  (\so,0) --  (z1) node[] {};
    \draw[fill=white] (\so+1.1,1.1) circle (1cm); 
    \foreach \i in {27,...,31}
	{
		\node [draw,rectangle,minimum width=2pt,minimum height=2pt,inner sep=0pt,fill](c\i) at ({2.7*sin(\i/23*pi r)+3+\so},{2.7*cos(\i/23*pi r)+3}){};
	}
	\draw (\so+1.74,1.7) circle (2pt); 
	\draw (\so+1.54,1.85) circle (2pt); 
	\centerarc[red,thick](3+\so,3)(204:246:2.3)
	\centerarc[red,thick,dashed](\so,0)(24:66:2.3)
	\draw [->, thick] (2.5,1) to [out=40,in=150] node[auto] {$\phi$} (\so-0.5,1);
	\node  at (0.4,2) {$\R^n$};
	\node  at (4.2,2) {$\Upsilon$};
\end{tikzpicture}
		\vspace{0cm}\caption{The kernel trick transforms the data in a higher dimensional feature space where a linear regression is performed.}\vspace{-0.0cm}
		\label{fig:kernel}
	\end{center}
\end{figure}
The function~$K\colon\R^{n\times {n_\D}}\times \R^{n\times {n_\D}}\to\R^{{n_\D}\times {n_\D}}$ is called the Gram matrix~$K_{j,l}= k(X_{:, l},X_{:, j})$ with~$j,l\in\lbrace 1,\ldots,{n_\D}\rbrace$. Each element of the matrix represents the covariance between two elements of the training data~$X$. The vector-valued function~$\bm{k}\colon\R^n\times \R^{n\times {n_\D}}\to\R^{n_\D}$ calculates the covariance between the test input~$\bm{z}^*$ and the input training data~$X$  
\begin{align}
	\bm{k}(\bm{z}^*,X)\text{ with }k_j = k(\bm{z}^*,X_{:, j})
\end{align}
for all~$j\in\lbrace 1,\ldots,{n_\D}\rbrace$. The covariance function depends on a set of hyperparameters~$\Phi=\{\varphi_1,\ldots,\varphi_{n_h}\}$ whose number~$n_h\in\N$ and domain of parameters depend on the employed function. A comparison of the characteristics of the different covariance functions can be found in~\cite{bishop2006pattern}. The prediction of each component of~$\bm f(\bm{z}^*)$ is derived from the joint distribution~\cref{sec2:for:joint_dist} and, therefore, it is a Gaussian distributed variable. The conditional probability distribution for the~$i$-th element of the output is defined by the mean and the variance
\begin{align}
	\mean_i(\f\vert \bm{z}^*,\mathcal D)=&{m}(\bm{z}^*)+\bm{k}(\bm{z}^*,X)\tran {(K+\sigma_i^2 I)}^{-1}\notag\\
	&(Y_{:,i}-\m(X))\label{for:gp_meanvar}\\
	\var_i(\f\vert \bm{z}^*,\!\mathcal D)=&k(\bm{z}^*,\bm{z}^*)\!-\!\bm{k}(\bm{z}^*\!,X)^\top{(K+\sigma_i^2 I)}^{-1}\notag\\
	&\bm{k}(\bm{z}^*\!,X).
\end{align}%
\begin{rem}
The existence of the inverse Gram matrix is essential for the prediction step. The Gram matrix is invertable if all vectors in the feature space~$\phi(X_{:,1}),\ldots,\phi(X_{:,{n_\D}})$ are independent. If there exist an~$i,j\in\N$ such that~$\phi(X_{:,i})=\phi(X_{:,j})$ or the number of training points exceeds the dimensionality of the feature space, i.e~${n_\D}>dim(\Upsilon)$, the condition can be violated. In this case, the Moore–Penrose-pseudoinverse is used. For further discussion on regularization methods for GPs, see~\cite{mohammadi2016analytic}.
\end{rem}
The~$n_f$ normally distributed components of~$\bm f\vert \bm{z}^*,\mathcal D$ are combined into a multi-variable Gaussian distribution 
\begin{align}
		\bm f\vert \bm{z}^*,\mathcal D &\sim \mathcal{N} (\bm\mean(\cdot),\Var(\cdot))\notag\\
		\bm \mean(\bm f\vert \bm{z}^*,\mathcal D)&=[\mean(f_1\vert \bm{z}^*,\mathcal D),\ldots,\mean(f_{n_f}\vert \bm{z}^*,\mathcal D)]\tran\label{for:multigp}\\
		\Var(\f\vert \bm{z}^*,\mathcal D)&=\diag(\var(f_1\vert \bm{z}^*,\mathcal D),\ldots,\var(f_{n_f}\vert \bm{z}^*,\mathcal D))\notag,
\end{align} 
where~$\Phi_i=\{\varphi_1,\ldots,\varphi_{n_h}\}$ is the set of hyperparameters for the~$i$-th output dimension. The hyperparameters are typically optimized by means of the likelihood function, thus by~$\varphi^{\{j\}}_i = \arg\max_{\varphi^{\{j\}}} \log P(Y_{:,i}|X,\varphi^{\{j\}})$ for all~$i\in\{1,\ldots,m\}$ and~$j\in\{1,\ldots,n_h\}$. A gradient based algorithm is often used to find at least a local maximum of the likelihood function~\cite{rasmussen2006gaussian}.
\begin{rem}
	Also the correlation between the dimensions of the state variable can be considered, e.g., by placing a separate covariance functions on the GP outputs~\cite{rakitsch2013all} or by using a multiple-output covariance function~\cite{melkumyan2011multi}.
\end{rem}
\subsection{Gaussian Process Dynamical Models}
\label{sec2:sec:GPDM}
Black-box models of nonlinear systems can be classified in many different ways. One main aspect of GPDMs is to distinguish between recurrent structures and non-recurrent structures. A model is called recurrent if parts of the regression vector depend on the outputs of the model. Even though recurrent models become more complex in terms of their behavior, they allow to model sequences of data, see~\cite{sjoberg1995nonlinear}. If all states are fed back from the model itself, we get a simulation model, which is a special case of the recurrent structure. The advantage of such a model is its property to be independent of the real system's states. Thus, it is suitable for simulations, as it allows multi-step ahead predictions without the need of the real system. In this article, we focus on two often-used recurrent structures: the Gaussian process state space model (GP-SSM) and the Gaussian process nonlinear error output (GP-NOE) model.
\subsubsection{Gaussian Process State Space Models}
Gaussian process state space models are structured as a discrete-time system. In this case, the states are the regressors, which is visualized in~\cref{fig:MS_SSM}. 
This approach allows to be more efficient, since the regressors are less restricted in their internal structure. Thus, a very efficient model in terms of number of regressors might be possible. The mapping from the states to the output can often be assumed to be known. The situation, where the output mapping describes a known sensor model, is such an example. It is mentioned in~\cite{frigola2013bayesian} that using too flexible models for both -~$\f$ and the output mapping - can result in problems of non-identifiability. Therefore, we focus on a known output mapping. The mathematical model of the GP-SSM is thus given by
\begin{align}
	\xkp&=\f(\vxi_t)=\begin{cases} 
f_1(\vxi_t)\sim \GP\left(m^1(\vxi_t),k^1(\vxi_t,\vxi_t^\prime)\right)\\
\vdots\hspace{0.9cm}\vdots\hspace{0.5cm}\vdots\\
f_{n_x}(\vxi_t)\sim \GP\left(m^{n_x}(\vxi_t),k^{n_x}(\vxi_t,\vxi_t^\prime)\right).
\end{cases}\notag\\
	\yk&\sim p(\yk\vert\xk,\bm{\gamma}_y),\label{for:gp_ssm}
\end{align}
where~$\vxi_t\in\R^{n_\xi},n_\xi=n_x+n_u$ is the concatenation of the state vector~$\xk\in\X\subseteq\R^{n_x}$ and the input~$\uk\in\mathcal{U}\subseteq\R^{n_u}$ such that~$\vxi_t=[{\xk};\vu_t]$. The mean function is given by continuous functions~$m^1,\ldots,m^{n_x}\colon\R^{n_\xi}\to\R$. The output mapping is parametrized by a known vector~$\bm{\gamma}_y\in\R^{n_\gamma}$ with~$n_\gamma\in\N$. The system identification task for the GP-SSM mainly focuses on~$\f$ in particular. It can be described as finding the state-transition probability conditioned on the observed training data.
\begin{rem}
	The potentially unknown number of regressors can be determined using established nonlinear identification techniques as presented in~\cite{keviczky1999nonlinear}, or exploiting embedded techniques such as automatic relevance determination~\cite{kocijan2016modelling}.
\end{rem}
\begin{figure}[bht]
	\begin{center}
		\tikzsetnextfilename{model_structure_ssm}
		\begin{tikzpicture}[node distance=2.5cm,auto,>=latex]
	\def\nodedist{2.5cm}
	\def\muxerheight{2cm}
	\tikzstyle{mux}=[draw, fill=black, minimum height=\muxerheight,minimum width=0.05cm,inner sep=0pt]
	\tikzstyle{block}=[draw, minimum height=0.5cm,minimum width=0.5cm,inner sep=2pt]
	\tikzstyle{textblock}=[text width=0.5cm, anchor=center,align=center,inner sep=0pt]
    \node [mux] (muxer1) {};
    \path (muxer1.west)+(-\nodedist,-\muxerheight/4) node (input_u) [coordinate] {};
    \path (muxer1.west)+(-\nodedist/2.5,\muxerheight/4) node (q) [block] {$q^{-1}$};
    \path (q.west)+(-\nodedist/8,0) node (qc) [coordinate] {};
    \path (muxer1.east)+(\nodedist/30,0) node (cord1) [coordinate] {};
    \path (muxer1.east)+(\nodedist/4,\muxerheight/8*3) node (GP1) [block] {$\GP$};
    \path (muxer1.east)+(\nodedist/4,\muxerheight/8*1) node (GP2) [block] {$\GP$};
    \path (muxer1.east)+(\nodedist/4,-\muxerheight/8*1) node (text1) [textblock] {\rvdots};
    \path (muxer1.east)+(\nodedist/4,-\muxerheight/8*3) node (GPn) [block] {$\GP$};
    \path (muxer1.east)+(\nodedist/2,0) node (muxer2) [mux] {};
    \path (muxer2.east)+(\nodedist/2,0) node (output_x) [coordinate] {};
    \path (output_x)+(\nodedist/2,0) node (GPy1) [block] {$p(\ykp\vert\xkp,\bm{\gamma}_y)$};
    \path (GPy1.east)+(\nodedist/2,0) node (output_y) [coordinate] {};
    \path[->] (q) edge node [pos=.4,yshift=-1mm] {$\xk$} (muxer1.west |- q.east);
    \path[->] (input_u) edge node [pos=.1,yshift=-1mm] {$\uk$} (muxer1.west |- input_u);
    \draw[->] (muxer1) -- (cord1) |- (GP1);
    \draw[->] (muxer1) -- (cord1) |- (GP2);
    \draw[->] (muxer1) -- (cord1) |- (GPn);
    %
    %
    \path[->] (GP1) edge (muxer2.west |- GP1);
    \path[->] (GP2) edge (muxer2.west |- GP2);
    \path[->] (GPn) edge (muxer2.west |- GPn);
    \draw[->] (muxer2) -- node [pos=.5,yshift=-1mm] {$\xkp$} (GPy1);
    \draw[->] (GPy1) -- node [pos=.5,yshift=-1mm] {$\ykp$} (output_y);
    \draw[->] (GPy1) -| node[below] (a) {} +(-\nodedist*0.95,\muxerheight/1.8) -| node[below] (b) {} (qc) -- (q); 
    \node[draw,dashed,inner xsep=0.1cm,inner ysep=0.1cm,fit=(muxer1) (muxer2) (q) (a) (qc) (b) (GPy1)] (box1)  {};
\end{tikzpicture}
		\caption{Structure of a GP-SSM with~$q$ as backshift operator, such that~$q^{-1}\xkp=\xk$.}
		\label{fig:MS_SSM}
	\end{center}
\end{figure}
\subsubsection{Gaussian Process Nonlinear Output Error Models}
The GP-NOE model uses the past~$n_\text{in}\in\N_{>0}$ inputs~$\uk\in\mathcal{U}\subseteq\R^{n_u}$ and the past~$n_\text{out}\in\N_{>0}$ output values~$\yk\in\R^{n_y}$ of the model as the regressors.~\Cref{fig:MS_NOE} shows the structure of GP-NOE, where the outputs are fed back. In combination with neural networks, this model type is also known as \emph{parallel model}.
The mathematical model of the GP-NOE is given by
\begin{align}
\begin{split}
	&\ykp=\h(\zk)=\begin{cases} 
h_1(\zk)\sim \GP\left(m^1(\zk),k^1(\zk,\zk^\prime)\right)\\
\vdots\hspace{0.9cm}\vdots\hspace{0.5cm}\vdots\\
h_{n_y}(\zk)\sim \GP\left(m^{n_y}(\zk),k^{n_y}(\zk,\zk^\prime)\right),
\end{cases}
\end{split}\label{for:gp_noe}
\end{align}
where the vector~$\zk\in\R^{n_\zeta},n_\zeta=n_\text{out} n_y+n_\text{in} n_u$ is the concatenation of the past output values~$\yk$ and input values~$\uk$ such that~$\zk=[\y_{t-n_\text{out}+1};\ldots;\y_t;\vu_{t-n_\text{in}+1};\ldots;\vu_t]$. 
The mean function is given by continuous functions~$m^1,\ldots,m^{n_y}\colon\R^{n_\zeta}\to\R$. In contrast to nonlinear autoregressive exogenous models, that focus on one-step ahead prediction, a NOE model is more suitable for simulations as it considers the multi-step ahead prediction~\cite{nelles2013nonlinear}. However, the drawback is a more complex training procedure that requires a nonlinear optimization scheme due to their recurrent structure~\cite{kocijan2016modelling}.
\begin{figure}[t]
	\begin{center}
		\tikzsetnextfilename{model_structure_noe}
		\vspace{0.2cm}
		\begin{tikzpicture}[node distance=2.5cm,auto,>=latex]
	\def\nodedist{2.5cm}
	\def\muxerheight{2cm}
	\def\muxerheightb{2cm}
	\tikzstyle{mux}=[draw, fill=black, minimum height=\muxerheight,minimum width=0.05cm,inner sep=0pt]
	\tikzstyle{block}=[draw, minimum height=0.5cm,minimum width=0.5cm,inner sep=2pt]
	\tikzstyle{textblock}=[text width=0.5cm, anchor=center,align=center,inner sep=0pt]
    \node [mux,minimum height=\muxerheightb] (muxer1) {};
    \path (muxer1.west)+(-\nodedist*2.1,-\muxerheightb/10) node (input_u) [coordinate] {};
    \path (muxer1.west)+(-\nodedist,-\muxerheightb/10*2.5) node (text2) [textblock] {\rvdots};
    \path (muxer1.west)+(-\nodedist*2.1,-\muxerheightb/10*4) node (input_un) [coordinate] {};
    \path (muxer1.west)+(-\nodedist/0.8,\muxerheightb/10*4) node (q) [block] {$q^{-\!1}$};
    \path (muxer1.west)+(-\nodedist,\muxerheightb/10*2.5) node (text3) [textblock] {\rvdots};
    \path (muxer1.west)+(-\nodedist/1.4,\muxerheightb/10) node (qn) [block] {$q^{-\!n_\text{out}}$};
    \path (q.west)+(-\nodedist/8,0) node (qc) [coordinate] {};
    \path (muxer1.east)+(\nodedist/12,0) node (cord1) [coordinate] {};
    \path (muxer1.east)+(\nodedist*0.3,\muxerheight/8*3) node (GP1) [block] {$\GP$};
    \path (muxer1.east)+(\nodedist*0.3,\muxerheight/8*1) node (GP2) [block] {$\GP$};
    \path (muxer1.east)+(\nodedist*0.3,-\muxerheight/8*1) node (text1) [textblock] {\rvdots};
    \path (muxer1.east)+(\nodedist*0.3,-\muxerheight/8*3) node (GPn) [block] {$\GP$};
    \path (muxer1.east)+(\nodedist*0.6,0) node (muxer2) [mux] {};
    \path (muxer2.east)+(\nodedist/2,0) node (output_x) [coordinate] {};
    \path[->] (q) edge node [yshift=-1mm,pos=.3,text width=0.5cm,align=left] {$\yk$} (muxer1.west |- q.east);
    \path[->] (qn) edge node [yshift=-1.2mm,pos=.25,text width=0.5cm,align=left] {$\y_{t\!-\!n_\text{out}\!+\!1}$} (muxer1.west |- qn.east);
    \path[->] (input_u) edge node [yshift=-1mm,pos=.05,text width=0.5cm,align=left] {$\uk$} (muxer1.west |- input_u);
    \path[->] (input_un) edge node [yshift=-1mm,pos=.05,text width=0.5cm,align=left] {$\vu_{t\!-\!n_\text{in}\!+\!1}$} (muxer1.west |- input_un);
    \draw[->] (muxer1) -- (cord1) |- (GP1);
    \draw[->] (muxer1) -- (cord1) |- (GP2);
    \draw[->] (muxer1) -- (cord1) |- (GPn);
    \path[->] (GP1) edge (muxer2.west |- GP1);
    \path[->] (GP2) edge (muxer2.west |- GP2);
    \path[->] (GPn) edge (muxer2.west |- GPn);
    \draw[->] (muxer2) -- node [yshift=-1mm,pos=.6] {$\ykp$} (output_x);
    \draw[->] (output_x) -| node[below] (a) {} +(-\nodedist*0.45,\muxerheightb/1.6) -| node[below] (b) {} (qc) -- (q); 
    \draw[->] (qc) |- (qn);
    \node[draw,dashed,inner xsep=0.05cm,inner ysep=0.1cm,fit=(muxer1) (muxer2) (q) (a) (qc) (b)] {};
\end{tikzpicture}
		\vspace{0cm}\caption{Structure of a GP-NOE model with~$q$ as backshift operator, such that~$q^{-1}\ykp=\yk$.}\vspace{-0.2cm}
		\label{fig:MS_NOE}
	\end{center}
\end{figure}
\begin{rem}
It is always possible to convert an identified input-output model into a state-space model, see~\cite{phan1970relationship}. However, focusing on state-space models only would preclude the development of a large number of useful identification result for input-output models.
\end{rem}
\subsubsection{Training}
	As the article focuses on the properties of GPDMs, regardless of the training procedure, we assume for the remainder of the paper that a training set~$\mathcal{D}$ is existent and available. In case of a GP-SSM, the training set consists of $X=[\vxi_0,\ldots,\vxi_{n_\D-1}]$ with $\vxi_t=[\x_t^\top,\uk^\top]^\top$ as input data, and $Y=[\x_1,\ldots,\x_{n_\D}]^\top$ as output data. In contrast, the training set of a GP-NOE model consists of $X=[\z_0,\ldots,\z_{n_\D-1}]$ with $\z_t=[\yk^\top,\uk^\top]^\top$ as input data, and $Y=[\y_1,\ldots,\y_{n_\D}]^\top$ as output data. For the sake of completeness, we present a simple way how to collect training data in the following. In case of available state measurements, the training data set for a GP-SSM can be directly created by recording the states and inputs as depicted in~\cref{fig:TD}. Here, the discrete-time system to model is operated by an arbitrary controller.
	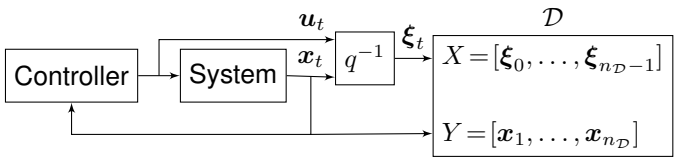
\begin{figure}[b]
	\begin{center}
	 \tikzsetnextfilename{training_bsb}	
	 \vspace{-0.3cm}
	 	\begin{tikzpicture}[auto, node distance=3cm,>=latex']
\tikzstyle{block} = [draw, fill=white, rectangle, minimum height=2em, minimum width=1em]
\tikzstyle{sum} = [draw, fill=white, circle, node distance=1cm]
\tikzstyle{input} = [coordinate]
\tikzstyle{output} = [coordinate]
\tikzstyle{t_output} = []
\tikzstyle{pinstyle} = [pin edge={to-,thin,black}]

    \node [input, name=input] {};
    \node [block, right of=input,node distance=2cm] (controller) {$\textsf{Controller}$};
    \node [block, right of=controller,node distance=2.15cm] (robot) {$\textsf{System}$};
    \node [block, right of=robot,node distance=1.75cm,yshift=2.3mm] (hatrobot) {$q^{-1}$};
    \node [block, right of=hatrobot,node distance=2.5cm,yshift=-5mm,label=$\D$,text width=8.4em] (data) {$X\!=\![\vxi_0,\ldots,\vxi_{n_\D-1}]$\\\vspace{0.6cm} $Y\!=\![\x_1,\ldots,\x_{n_\D}]$};
    
	\draw [->] (controller) -- node[name=t,below] {} (robot);
    \draw [->] (robot.east) -- node[name=q] {$\x_t$} ([yshift=-2.5mm]hatrobot.west);
    
    \coordinate [below of=q,node distance=1cm] (dp) {};
    \draw [->] (q) |- (dp) -| (controller);

    \draw [->] (t.north) |- node[pos=0.935] {$\bm{u}_t$} ([yshift=2.5mm]hatrobot.west);
    \draw [->] (hatrobot.east) -- node[] {$\vxi_t$} ([yshift=5mm]data.west);
    \draw [->] (dp) -- ([yshift=-5mm]data.west);

\end{tikzpicture}
		\caption{Block diagram of the generation of the training data set $\mathcal D$ with $n_\D$ data points for GP-SSMs.\label{fig:TD}}
	\end{center}
\end{figure}
	The only condition on the controller is that a finite sequence of training data of the system can be collected whereas stability is not necessarily required. As stated in~\cref{sec2:sec:GPDM}, the transition from the state $\xk$ to the output $\yk$ is assumed to be known.\\
	However, if the state is intractable, well-known methods for the training are based on variational inference and the introduction of inducing points, see~\cite{wang2008gaussian,frigola2013bayesian,eleftheriadis2017identification} for more details. More information about the training procedure of GP-NOE models are presented in~\cite{kocijan2005dynamic,kocijan2011output}.
	\begin{rem}
	In this article, the training input and output data is always denoted by $X$ and $Y$, respectively, to be in line with the standard notation. Note that in case of GP-SSMs, the set $Y$ does not contain the outputs $\yk$ but the next states ahead.
	\end{rem}
\subsection{The crux of simulation}
\label{sec:sim}
The prediction with discrete-time GPDMs, needed for simulations and model-based control approaches, is more challenging than GP prediction: The reason is the feedback of the model's output to the input that manifests as correlation between the current and past states defined by the GP model. Therefore, a prediction with the presented GP models in~\cref{sec2:sec:GPDM} would require the sampling of the probabilistic mappings~$\dyn$ and~$\h$ given by~\cref{for:gp_ssm} and~\cref{for:gp_noe}, respectively. Once sampled, the model could be treated as standard discrete-time system. Unfortunately, these functions are defined on the sets~$\X\subseteq\R^{n_x},\mathcal{U}\subseteq\R^{n_u}$ which contain infinitely many points. Thus, it would be necessary to draw an infinite-dimensional object which represents a sample of the probabilistic mappings. This is not possible without further simplifications, e.g., discrete sampling with interpolation, see~\cite{umlauft:ecc2018}. To overcome this issue, the probabilistic mapping is marginalized out to respect the nonparametric nature of the model, see~\cite{matthews2016sparse} for more details. The result is a probability distribution of the states without dependencies on the probabilistic mappings~$\dyn,\h$. However, the marginalization of~$\dyn,\h$ leads to dependencies across time for the states. For the prediction of the next state ahead $\xkp$, the nature of GP models allows to include the past states as noise-free "training data" in a way that there exists an analytic closed-form, see \cite{frigola2016bayesian}. We formally restate the non-Markovian prediction as given in~\cite{matthews2016sparse} in the next property.
\begin{rem}
    For the sake of notational simplicity, we consider GPDMs with identical kernels and identical noise of the training data for each output dimension. The results can easily be extended to GPDMs with different kernels and noise for each output dimension.
\end{rem}
\begin{propy}\label{propy:gpssm}
    Consider a GP-SSM~\cref{for:gp_ssm} with training set~$\D=\{X,Y\}$, where~$Y$ is corrupted by Gaussian noise~$\mathcal{N}(0,\sigma_n^2 I)$. Then, the conditional distribution of the next state ahead~$\xkp\in\R^{n_x}$ and output~$\y_{t+1}\in\R^{n_y}$ is given by
    \begin{align}
    \begin{split}
    	\xkp\vert \vxi_{0:t},\D&\sim\mathcal{N}\big(\Mean(\xkp\vert\vxi_{0:t},\mathcal{D}),\Var(\xkp\vert\vxi_{0:t},\mathcal{D})\big)\\
    	\mean_i(\cdot)&=m(\vxi_t)+\bm{k}(\vxi_t, X_t)^\top {K}_t^{-1}\big([ {Y_t}]_{:,i}-\bm{m}( X_t)\big)\\
    	\Var_{i,i}(\cdot)&=k(\vxi_t,\vxi_t)-\bm{k}(\vxi_t, X_t)^\top {K}_t^{-1} \bm{k}(\vxi_t, X_t)\\
    	p(\y_{t+1}\vert \vxi_{0:t}&,\bm{\gamma}_y,\D)=p(\y_{t+1}\vert \x_{t+1},\bm{\gamma}_y)p(\x_{t+1}\vert \vxi_{0:t},\D)
    	\end{split}\label{for:gpssm_pred}
    \end{align}
    with~$\x_0\in\R^{n_x}$ for all~$t\geq 0$ and extended data matrices~$ X_t\in\R^{n_\xi\times({n_\D}+t)}$,~$ Y_t\in\R^{({n_\D}+t)\times n_x}$
    \begin{align}
    	 X_t&=X,& Y_t&=Y & &\text{if } t=0\notag\\ 
    	 X_t&=[X,\vxi_{0:t-1}],& Y_t&=[Y^\top,\x_{1:t}]^\top & &\text{otherwise}. \label{sec3:for:gpssmindex1}
    \end{align}
    The Gram matrix~${K}_t\in\R^{({n_\D}+t)\times({n_\D}+t)}$ is defined as
    \begin{align}
    {K}_t=\begin{cases}
    	\begin{bmatrix}
    		K+\sigma_n^2 I & K(\vxi_{0:t-1},X)\\
    		K(\vxi_{0:t-1},X)^\top & K(\vxi_{0:t-1},\vxi_{0:t-1})
    	\end{bmatrix}, & \text{if } t>0\\
    	K+\sigma_n^2 I, & \text{otherwise}.
    \end{cases}\label{for:gram}
    \end{align}
\end{propy}
\begin{rem}
    \Cref{propy:gpssm} shows that dependency across time appears through past states and inputs treated as noise-free "training data". However, it should not be mistaken as real training data, which is also included in the Gram matrix $K_t$, but seen as a way to include the correlation through time.
\end{rem} 
For the first step, i.e~$t=0$, the conditional distribution as given in~\cref{propy:gpssm} is identical to the standard GP prediction with predicted mean and variance given by~\cref{for:gp_meanvar}. For~$t>0$, the current state is feedbacked to the input, as shown in~\cref{fig:MS_SSM}. Using the joint Gaussian distribution property~\cref{sec2:for:joint_dist} of the GP, we obtain the joint distribution 
\begin{align}
	&\begin{bmatrix} ( Y_t)_{:,i} \\ (x_{t+1})_i \end{bmatrix}\sim \mathcal{N} \left(\begin{bmatrix}\bm m(X)\\\bm m(\vxi_{0:t})\end{bmatrix},
	\begin{bmatrix}  {K}_t 		&	{K^\prime}_t\\
					{K^\prime}^\top_t	&	k(\vxi_t,\vxi_t) \end{bmatrix}\right),\label{sec3:for:jdproof}
\end{align}
where~${K^\prime}^\top_t=\left[\bm{k}(\vxi_t,X)^\top, \bm{k}(\vxi_t,\vxi_{0:t-1})^\top\right]$. Based on~\cref{sec3:for:jdproof}, the conditional probability distribution of the next state ahead~$\xkp$ is computed.
\begin{rem}
    If we consider a state feedback law $\inputu_t=\bm{g}(\x_t)$ with $\bm{g}\colon\R^{n_x}\to\R^{n_u}$, the extended input vector is then given by $\vxi_t=[\x_t^\top,\bm{g}(\x_t)^\top]^\top$ and \cref{propy:gpssm} can also be used to sample trajectories for closed-loop simulations.
\end{rem} 
Analogously, we introduce the prediction for the GP-NOE model.
\begin{propy}\label{thm:gpnoe}
Consider a GP-NOE model~\cref{for:gp_noe} with training set~$\D=\{X,Y\}$, where the output data~$Y$ is corrupted by Gaussian noise~$\mathcal{N}(0,\sigma_n^2 I)$. Then, the conditional distribution of the next output~$\ykp\in\R^{n_y}$ is given by
\begin{align}
\begin{split}
	\ykp\vert \z_{0:t},\D&\sim\mathcal{N}\big(\Mean(\ykp\vert\z_{0:t},\mathcal{D}),\Var(\ykp\vert\z_{0:t},\mathcal{D})\big)\\
	\mean_i(\cdot)&=m(\zk)+\bm{k}(\zk, X_t)^\top {K}_t^{-1}\big([ {Y_t}]_{:,i}-\bm{m}( X_t)\big)\\
	\Var_{i,i}(\cdot)&=k(\zk,\zk)-\bm{k}(\zk, X_t)^\top {K}_t^{-1} \bm{k}(\zk, X_t)
	\end{split}\label{for:gpnoe_pred}
\end{align}
with~$\z_{0}\in\R^{n_\zeta}$ for all~$t\geq 0$ and the extended data matrices~$ X_t\in\R^{n_\zeta\times({n_\D}+t)}$,~$ Y_t\in\R^{({n_\D}+t)\times n_y}$
\begin{align}
	 X_t&=X,& Y_t&=Y & & \text{if } t=0\notag\\ 
	 X_t&=[X,\z_{0:t-1}],& Y_t&=[Y^\top,\y_{1:t}]^\top & &\text{otherwise}.
\end{align}
The Gram matrix~${K}_t\in\R^{({n_\D}+t)\times({n_\D}+t)}$ is defined as
\begin{align}
{K}_t=\begin{cases}
	\begin{bmatrix}
		K+\sigma_n^2 I & K(\z_{0:t-1},X)\\
		K(\z_{0:t-1},X)^\top & K(\z_{0:t-1},\z_{0:t-1})
	\end{bmatrix}, & \text{if } t>0\\
	K+\sigma_n^2 I, & \text{otherwise}.
\end{cases}\label{for:gramnoe}
\end{align}
\end{propy}
\subsection{Need for Markovian models}
The previous section shows that the next step ahead state~$\xkp$ of a GP-SSM is a sample drawn from a Gaussian distribution with the posterior mean and variance based on the previous states and inputs. This leads to dependencies between the states such that the dynamical model loses the Markov property, i.e.,~$\xkp$ depends not only on~$\xk$ but on all previous states~$\x_{0:t}$. The resulting issues are from theoretical and practical nature as presented in the following.
\subsubsection{Theoretical issues}
The proof of system properties such as stability, convergence rate, and performance is key for mainly all applications in control systems - especially in safety critical environments as, e.g., autonomous driving. Over the last decades, the control community has developed a large amount of tools for the analysis and control synthesis for dynamical systems. However, these tools mainly focus on systems with Markov property. For instance, the Lyapunov stability is a standard concept for the analysis of nonlinear open- and closed-loop systems which needs the Markov property. For non-Markovian systems, the amount of tools are significantly decreased such that the analysis of control systems with GPDMs is challenging.
\subsubsection{Practical issues}
As the past states are treated as new "training points" without noise, the size of the covariance matrix $K_t$ increases with each time step, see~\cref{propy:gpssm}. This results not only in a strongly increasing computing time for the prediction but also an intractable memory problem for long time simulations. Even though recent findings allows the computing time to be at least linear, e.g. see~\cite{wilson2020efficiently}, the control related implications of these approximations are still unaddressed.
\section{Approximated Models}
\label{sec:approx}
To overcome the issue of the non-Markovian property of GPDMs, we introduce approximated GPDMs where only a subset of previous states and inputs are considered for the prediction. These models allow to use all analyse tools available for Markovian systems and to keep the computation time constant. First, we introduce the formal description of this approximated model. For this purpose, we define the matrix~$\Xi^m_t\in\R^{n_\xi\times \ubar{m}}$ consisting of past states and inputs as
\begin{align}
	\Xi^m_t&\coloneqq\begin{cases}
			\emptyset & \text{if }\overline{m}=0\vee t=0\\
	[\vxi_{t-1},\ldots,\vxi_{t-\ubar{m}}] & \text{otherwise,}\\ 
	\end{cases}\label{for:memgpssm}
\end{align}
which are used for the prediction. The \emph{maximum length of memory}~$\overline{m}\in\N$ defines how many past states and inputs are considered for the prediction of the next state. The resulting \emph{actual length of memory}~$\ubar{m}=\min(t,\overline{m})$ is the number of states and inputs which are actually available. The actual length and the maximum length only differ if the number of past states beginning with~$\x_0$ is less than~$\overline{m}$. The prediction of the next state ahead and the output~$\ykp\in\R^{n_y}$ is given by
\begin{align}	
		\xkp^m&\sim \ND\big(\underbrace{\Mean(\xkp^m\vert\vxi_t,\Xi^m_t,\mathcal{D})}_{\dyn_t(\vxi_t,\Xi^m_t)},\underbrace{\Var(\xkp^m\vert\vxi_t,\Xi^m_t,\mathcal{D})}_{F_t(\vxi_t,\Xi^m_t)}\big)\notag\\
		\ykp\vert\xkp^m&\sim p(\ykp\vert\xkp^m,\bm{\gamma}_y).\label{for:gp_ssm_m}
\end{align} 
For simplicity in the notation, we introduce two helper functions~$\dyn_t\!\colon\!\R^{n_\xi}\!\times\!\R^{n_\xi\times \ubar{m}}\!\to\!\R^{n_x}$ and~$F_t\!\colon\!\R^{n_\xi}\!\times\!\R^{n_\xi\times \ubar{m}}\!\to\!\R^{n_x\times n_x}$. The posterior mean and variance of the~$i$-th element of~$\xkp^m$ is given by
\begin{align}
\begin{split}
&f_t(\vxi_t,\Xi^m_t)_i\!=\!m(\vxi_t)\!+\!\bm{k}(\vxi_t, X^m_t)^\top \!{({K}^m_t)}^{-1}\!\big([ Y_t^m]_{:,i}\!-\!\bm{m}({X}^m_t)\!\big)\\
&F_t(\vxi_t,\Xi^m_t)_{i,i}\!=k(\vxi_t,\vxi_t)-\bm{k}(\vxi_t, X^m_t)^\top {({K}^m_t)}^{-1} \bm{k}(\vxi_t,{X}^m_t),
\end{split}\label{sec3:for:gpassmmeanvar}
\end{align}
respectively. The extended data matrices~$ X^m_t\in\R^{n_\xi\times({n_\D}+\ubar{m})}$ and~$ Y^m_t\in\R^{({n_\D}+\ubar{m})\times n_y}$ are denoted by
\begin{align}
	 X^m_t&=X,& Y^m_t&=Y & &\hspace{-0.9cm}\text{if } \overline{m}\!=\!0\vee t\!=\!0\notag\\ 
	 X^m_t&=[X,\vxi_{t-\overline{m}:t-1}],& Y^m_t&=[Y^\top,\x_{t-\overline{m}+1:t}]^\top & &\text{otherwise}\label{sec3:for:gpssmindex2}
\end{align}
where only elements back to $t=0$ in case of negative $t-\overline{m}$ are considered. The corresponding Gram matrix~${K}^m_t\in\R^{({n_\D}+\ubar{m})\times({n_\D}+\ubar{m})}$ is given by
\begin{align}
	{K}^m_t=&\begin{bmatrix}
		K(X,X)+\sigma_n^2 I & K(\vxi_{t-\overline{m}:t-1},X)\\
		K(\vxi_{t-\overline{m}:t-1},X)^\top & K(\vxi_{t-\overline{m}:t-1},\vxi_{t-\overline{m}:t-1})
	\end{bmatrix}\label{for:gramm}\\
	&\text{if }t>0\wedge \overline{m}>0 \text{ and } K(X,X)+\sigma_n^2 I \text{ otherwise}.\notag
\end{align}
Note that the prediction in~\cref{for:gp_ssm_m} is based on the past states and inputs back to the time step~$t-\overline{m}$, as indicated in~\cref{sec3:for:gpssmindex2}. In contrast, the prediction of a GP-SSM is based on the full history of states and inputs, see~\cref{sec3:for:gpssmindex2}.
\begin{defn}
\label{sec3:defn:gp_assm}
We call~\cref{for:gp_ssm_m} a \emph{Gaussian process approximated state space model} (GP-ASSM) with maximum memory length~$\overline{m}$.
\end{defn}
\begin{rem}
For~$\overline{m}=\infty$, the prediction depends on all past states, i.e.,
\begin{align}
	\xkp^\infty\sim \ND\big(\!\Mean(\xkp^\infty\vert\vxi_t,\ldots,\vxi_0,\mathcal{D}),\Var(\xkp^\infty\vert\vxi_t,\ldots,\vxi_0,\mathcal{D})\big)\label{sec3:for:gpassminf}
\end{align}
and thus, equals the true distribution in~\cref{for:gpssm_pred} without Markovian property. The most simple approximation is given for maximum memory length~$\overline{m}=0$
\begin{align}
	\xkp^0\sim \ND\big(\Mean(\xkp^0\vert\vxi_t,\mathcal{D}),\Var(\xkp^0\vert\vxi_t,\mathcal{D})\big),\label{sec3:for:gpassmnull}
\end{align}
where the next state ahead is independent of all past states except the current state and input~$\vxi_t$. GP-ASSMs with finite maximum length of memory~$\overline{m}$ are \emph{Markov chains of finite order} as they depend on a finite set of past states and input. 
\end{rem}
\Cref{fig:time_dep} visualizes the relation between actual length~$\ubar{m}$ and the maximum length~$\overline{m}$.
\begin{exam}
\label{sec3:exam:approx}
The idea of the presented approximation is visualized in the top plot of~\cref{fig:sr} by a one-dimensional GP-ASSM with maximum length of memory~$\overline{m}=0$. For the sake of simplicity, the external input is set to zero~$u_t=0$ for all~$t\in\N$. The distribution of the next state ahead depends only on the current state~$x^0_t$ as it is always sampled from a Gaussian distribution disregarding the history of the past states. Thus, for a given~$x^0_0$, the next state~$x^0_1$ (red circle) is sampled from a Gaussian distribution (green line), where the mean and variance are based on~$x^0_0$, see~\cref{sec3:for:gpassmnull}. In the next time step,~$x^0_2$ (red circle) is sampled from a Gaussian distribution (green line), where the mean and variance are solely based on~$x^0_1$. 
\begin{figure}[ht]
	\begin{center}
		\tikzsetnextfilename{time_dep}
		\vspace{0cm}
		\begin{tikzpicture}[node distance=1cm,auto,>=latex,font={\sffamily}]
	\tikzstyle{block}=[draw,circle, minimum height=1cm,minimum width=1cm,inner sep=1pt]
	\tikzstyle{textblock}=[text width=0.4cm, anchor=center,align=center,inner sep=0pt]
    \node [block] (x0) {$x_0$};
    \node [textblock, right of=x0] (d1) {$\cdots$};
    \node [block, right of=d1] (x2) {$x_{t-\overline{m}}$};
    \node [textblock, right of=x2] (d2) {$\cdots$};
    \node [block, right of=d2] (x3) {$x_{t-1}$};
    \node [block, right of=x3] (x4) {$x_{t}$};
    \node [block, right of=x4,node distance=1.7cm] (x5) {\footnotesize $x_{t+1}$};
    \draw[->] (x4) to (x5);
    \draw[->] (x3) to [out=40,in=145] (x5);
    \draw[->] (x2) to [out=45,in=130] (x5);
    \draw [decorate,decoration={brace,mirror,amplitude=7}] ([yshift=-5mm]x2.west) --node[below=3mm]{$\ubar{m}$ states in memory} ([yshift=-5mm]x3.east);
\end{tikzpicture}
		\vspace{-0cm}\caption{Time dependencies for the next step ahead state~$\xkp$ with the actual length of the memory~$\ubar{m}=\min(t,m)$ and maximum length $\overline{m}$.}\vspace{-0.5cm}
		\label{fig:time_dep}
	\end{center}
\end{figure}
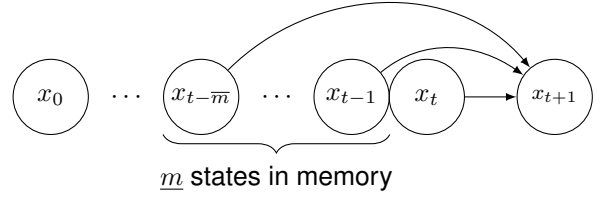
This procedure is continued for the following time steps. As the distribution (green line) of the next state~$x^0_{t+1}$ is independent of the past states~$x^0_{t-1},\ldots,x^0_0$, it is always equal to the distribution of the GP (mean and 2-sigma uncertainty) at state~$x^0_{t}$.
In contrast, the true sampling ($\overline{m}=\infty$) with a one-dimensional GP-SSM considers all past states~$x^\infty_t,\ldots,x^\infty_0$, see~\cref{sec3:for:gpassminf}. In~\cref{fig:sr1}, we start again with a given~$x_0^\infty$. The next state~$x_1^\infty$ (red circle) is sampled from a Gaussian distribution based on the initial state. Then,~$x_2^\infty$ is sampled based on~$x_1^\infty$ and~$x_0^\infty$. For this purpose, the pair~$(x_{0}^\infty,x_{1}^\infty)$ is added as noise free training data, see~\cref{sec3:for:gpssmindex1}. Thus, for any following state where~$x_{t}^\infty=x_{0}^\infty,t\geq 2$, the next state is given by~$x_{t+1}^\infty=x_{1}^\infty$. Due to the dependency on all past states, the distribution of states, which are not yet added as training data, differ from the predicted mean and variance of the GP. This is visualized at the distribution of~$x_4^\infty$ (green line) in contrast to the mean (blue line) and the 2-sigma uncertainty (gray shaded area) of the GP. This sampling procedure is necessary since the state mapping~$f$, given by~\cref{for:gp_ssm}, can not be drawn directly due to the definition over an infinite set~$\X\subseteq\R^{n_x}$. In~\cref{fig:sr1}, the mapping~$f$ is illustratively drawn (yellow line) over a finite but large number of states.
\begin{figure}[h]
	\begin{center}
		\tikzsetnextfilename{section3_scenario_vs_random}
		\captionsetup{type=figure}\begin{tikzpicture}
\begin{axis}[
    name=plot1,
  font={\sffamily},
  legend pos=north west,
  grid style={dashed,gray},
  grid = both,
       width=\columnwidth,
  height=5cm,
  ymin=3.5,
  ymax=13,
  xmin=2,
  xmax=10,
  legend style={font=\sffamily\footnotesize},
  legend cell align={left},
  legend pos={north west},
  xtick={4,4.73989425718806,5.84059997918927,7.7916654365068},
  ytick={4.73989425718806,5.84059997918927,7.79166543650687,10.0359453535049},
      xticklabels={$x_0^0$,$x_1^0$,$x_2^0$,$x_3^0$},
      yticklabels={$x_1^0$,$x_2^0$,$x_3^0$,$x_4^0$}]
  \addplot[color=blue,line width=2pt] table [x index=0,y index=1]{data/sz_approach_1.dat};
  \addplot+[name path=varp1, color=gray,opacity=0.3, no marks] table [x index=0,y expr=\thisrowno{1}+\thisrowno{2}]{data/sr_approach_1.dat};
\addplot+[name path=varm1, color=gray,opacity=0.3, no marks] table [x index=0,y expr=\thisrowno{1}-\thisrowno{2}]{data/sr_approach_1.dat};
\addplot[gray,opacity=0.5] fill between[ of = varm1 and varp1]; 
\addplot[only marks,color=black,line width=1pt,mark=+,mark size=4] table [x index=0,y index=1]{data/sr_approach_2.dat};
\addplot[color=ForestGreen,line width=2pt] table [x index=0,y index=4]{data/sr_approach_4.dat};
\addplot[color=ForestGreen,line width=2pt] table [x index=1,y index=5]{data/sr_approach_4.dat};
\addplot[color=ForestGreen,line width=2pt] table [x index=2,y index=6]{data/sr_approach_4.dat};
\addplot[color=ForestGreen,line width=2pt] table [x index=3,y index=7]{data/sr_approach_4.dat};
\addplot[only marks,color=red,line width=1pt,mark=o,mark size=4] table [x index=0,y index=1]{data/sr_approach_3.dat};
\legend{Mean,,,2-sigma,Data points,,,,,States}
\end{axis}
\end{tikzpicture} 
		\vspace{-0.1cm}
		\captionof{figure}{Sampling of a one-dimensional GP-ASSM with squared exponential kernel.}\vspace{-0.5cm}
		\label{fig:sr}
	\end{center}
\end{figure}
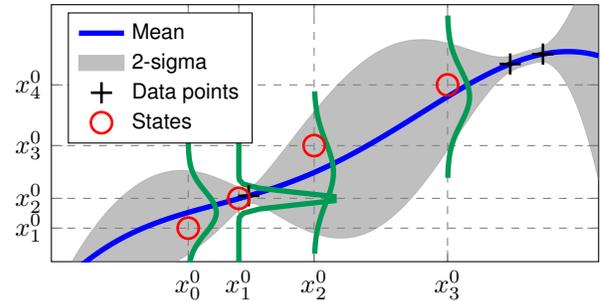
\begin{figure}[h]
	\begin{center}
		\tikzsetnextfilename{section3_scenario_vs_random1}
		\vspace{-0.3cm}
		\captionsetup{type=figure}\begin{tikzpicture}
\begin{axis}[
  legend pos=north west,
  grid style={dashed,gray},
  grid = both,
       width=\columnwidth,
       font={\sffamily},
  height=5cm,
  ymin=3.5,
  ymax=13,
  xmin=2,
  xmax=10,
  legend style={font=\sffamily\footnotesize},
  legend cell align={left},
  legend pos={north west},
  xtick={4,4.73989425718806,5.70651700149693,7.66700185417600},
  ytick={4.73989425718806,5.70651700149693,7.66700185417600,10.6806477332591},
  xticklabels={$x_0^\infty$,$x_1^\infty$,$x_2^\infty$,$x_3^\infty$},
  yticklabels={$x_1^\infty$,$x_2^\infty$,$x_3^\infty$,$x_4^\infty$}]
  \addplot[color=blue,line width=2pt] table [x index=0,y index=1]{data/sz_approach_1.dat};
  \addplot+[name path=varp1, color=gray,opacity=0.3, no marks] table [x index=0,y expr=\thisrowno{1}+\thisrowno{2}]{data/sz_approach_1.dat};
\addplot+[name path=varm1, color=gray,opacity=0.3, no marks] table [x index=0,y expr=\thisrowno{1}-\thisrowno{2}]{data/sz_approach_1.dat};
\addplot[gray,opacity=0.5] fill between[ of = varm1 and varp1]; 
\addplot[only marks,color=black,line width=1pt,mark=+,mark size=4] table [x index=0,y index=1]{data/sz_approach_2.dat};
\addplot[color=Goldenrod,line width=1pt] table [x index=0,y index=3]{data/sz_approach_1.dat};
\addplot[only marks,color=red,line width=1pt,mark=o,mark size=4] table [x index=0,y index=1]{data/sz_approach_3.dat};
\addplot[color=ForestGreen,line width=2pt] table [x index=3,y index=7]{data/sz_approach_4.dat};
\legend{,,,,,Sample of state mapping $f$}
\end{axis}
\end{tikzpicture} 
		\vspace{-0.1cm}
		\captionof{figure}{Sampling of a one-dimensional GP-SSM with squared exponential kernel.}\vspace{-0.5cm}
		\label{fig:sr1}
	\end{center}
	\end{figure}
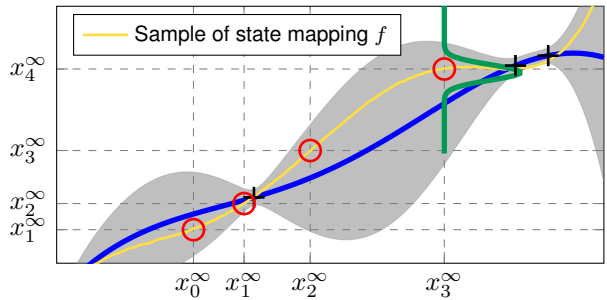
\end{exam}
\begin{rem}
More advanced strategies for defining the subset of states, that are stored in the memory, are also conceivable. Namely, the same methods as for sparsification of the training data can be exploited. For instance, approaches based on the effective prior~\cite{quinonero2005unifying} or pseudo-inputs~\cite{snelson2006sparse} have already been successfully applied for sparsification.
\end{rem}
In the following, we transfer this formal description to GP-NOE models. In comparison to GP-SSMs, the GP-NOE models do not have explicitly defined states. Therefore, we define the matrix of past outputs and inputs as
\begin{align}
	\Lambda^m_t&=\begin{cases}
		\emptyset & \text{if } \overline{m}=0\vee k=0\\
		[\z_{t-1},\ldots,\z_{t-\ubar{m}}] & \text{otherwise}\\ 
	\end{cases}
\end{align}
with~$\overline{m}\in\N$ defining the maximum length of memory and~$\ubar{m}=\min(t,m)$, the actual length of memory. The prediction of the next output~$\ykp\in\R^{n_y}$ is given by
\begin{align}	
		\ykp^m\sim\ND\big(\underbrace{\Mean(\ykp^m\vert\z_t,\Lambda^m_t,\mathcal{D})}_{\bm{h}_t(\z_t,\Lambda^m_t)},\underbrace{\Var(\ykp^m\vert\z_t,\Lambda^m_t,\mathcal{D})}_{H_t(\z_t,\Lambda^m_t)}\big).
	\label{for:gp_noe_m}
\end{align} 
For simplicity in the notation, we introduce the helper functions~$\h_t\colon\R^{n_\zeta}\times\R^{n_\zeta\times \ubar{m}}\to\R^{n_y}$ and~$H_t\colon\R^{n_\zeta}\times\R^{n_\zeta\times \ubar{m}}\to\R^{n_y\times n_y}$. The mean~$[h_t]_i$ and variance~$[H_t]_{i,i}$ of the~$i$-th output dimension is given by
\begin{align}
[h_t]_i=&m(\z_t)+\bm{k}(\z_t, X_t)^\top {({K}^m_t)}^{-1}([ Y_t]_{:,i}-\bm{m}({X}^m_t))\notag\\
[H_t]_{i,i}=&k(\z_t,\z_t)-\bm{k}(\z_t, X_t)^\top {({K}^m_t)}^{-1} \bm{k}(\z_t,{X}^m_t).
\end{align}
For GP-NOE models, we define the extended training sets~$ X^m_t\in\R^{n_\zeta\times({n_\D}+\ubar{m})}, Y^m_t\in\R^{({n_\D}+\ubar{m})\times n_y}$ as
\begin{align}
	 X^m_t&=X,& Y^m_t&=Y & &\hspace{-0.9cm}\text{if } \overline{m}\!=\!0\vee t\!=\!0\notag\\ 
	 X^m_t&=[X,\z_{t-\overline{m}:t-1}],& Y^m_t&=[Y^\top,\y_{t-\overline{m}+1:t}]^\top & &\text{otherwise}\notag
\end{align}
with the Gram matrix~${K}^m_t\in\R^{({n_\D}+\ubar{m})\times({n_\D}+\ubar{m})}$ as
\begin{align}
	{K}^m_t=&\begin{bmatrix}
		K(X,X)+\sigma_n^2 I & K(\z_{t-\overline{m}:t-1},X)\\
		K(\z_{t-\overline{m}:t-1},X)^\top & K(\z_{t-\overline{m}:t-1},\z_{t-\overline{m}:t-1})
	\end{bmatrix}\label{for:grammnoe}\\
	&\text{if }t>0 \wedge \overline{m}>0 \text{ and } K(X,X)+\sigma_n^2 I \text{ otherwise.}\notag
\end{align}
\begin{defn}
We call~\cref{for:gp_noe_m} a \emph{Gaussian process approximated nonlinear output error} (GP-ANOE) model with maximum memory length~$\overline{m}$.
\end{defn}
Having introduced the formal description for the approximations of the non-Markovian dynamics, we analyze the approximation error in the following.
\subsection{Approximation Error}
In this section, we present the computation of the error between the true state distribution~$\x_{t+1}$ given by~\cref{for:gpssm_pred} and the approximated distribution~$\xkp^m$ based on the maximum length of memory~$\overline{m}$. As the Kullback-Leibler (KL) divergence is an important measure of how one probability distribution differs from a second, we start with the KL divergence of the GP-SSM prediction from the GP-ASSM prediction. For the sake of clarity, we define the following notational simplifications 
\begin{align}
F_t^m&\coloneqq F_t(\vxi_t,\Xi^m_t),& F_t^\infty&\coloneqq F_t(\vxi_t,\Xi^\infty_t)\\
\dyn_t^m&\coloneqq\dyn_t(\vxi_t,\Xi^m_t),& \dyn_t^\infty&\coloneqq\dyn_t(\vxi_t,\Xi^\infty_t)
\end{align}
for the mean and variance given by~\cref{for:gp_ssm_m}.
\begin{prop}\label{prop:KL}
Consider a GP-ASSM with maximum length of memory~$\overline{m}\in\N$ and data set~$\mathcal{D}$ such that
\begin{align}
\x_{t+1}^m&\sim \ND\big(\Mean(\xkp\vert\vxi_t,\Xi^m_t,\mathcal{D}),\Var(\xkp\vert\vxi_t,\Xi^m_t,\mathcal{D})\big)\notag
\end{align}
with~$\x_0\in\R^{n_x}$. For given past states and inputs~$\vxi_{0:t}$, where~$\vxi_t\neq\vxi_0,\ldots,\vxi_{t-1}$, the KL-divergence of the true distribution~$\xkp$ from the approximation~$\xkp^m$ is given by
\begin{align}
	d_{\mathrm{KL}}(\xkp \| \xkp^m)&=\frac{1}{2}\Delta_t^\top [F_t^m]^{-1}\Delta_t-n_x+\tr\left({F_t^\infty}[F_t^m]^{-1}\right)\notag\\
	&+\ln\Big[\tr\left([F_t^\infty]^{-1}{F_t^m}\right)\!\Big]\label{for:KLprop}
\end{align}
with~$\Delta_t={\dyn_t^m}-{\dyn_t^\infty}$.
\end{prop}
\begin{proof}
 For given past states and inputs~$\vxi_{0:t}$, the next state~$\xkp$ of the GP-SSM and the next state~$\xkp^m$ of the GP-ASSM are Gaussian distributed such that the KL-divergence is given by
\begin{align}
	d_{\mathrm{KL}}(\xkp \| \xkp^m)\!=&\frac{1}{2}\bigg[ \!\tr\big([F_t^m]^{-1}{F_t^\infty}\big)\!+\!\big({\dyn_t^m}\!-\!{\dyn_t^\infty}\big)^\top\! [F_t^m]^{-1}\notag\\
	&\big({\dyn_t^m}\!-\!{\dyn_t^\infty}\big)-{n_x}+\ln\left(\frac{|{F_t^m}|}{|{F_t^\infty}|}\right)\bigg]\label{for:mse0}
\end{align}
using the definition of~$F_t,\dyn_t$ in~\cref{for:gp_ssm_m}. As the variance of each element in~$\xkp$ and~$\xkp^m$ is independent, see~\cref{sec3:for:gpassmmeanvar}, the KL-divergence can be rewritten to
\begin{align}
	d_{\mathrm{KL}}(\xkp \| \xkp^m)&=\frac{1}{2}\sum_{i=1}^{n_x}\bigg[\frac{[{F_t^\infty}]_{i,i}+([{f_t^m}]_i-[{f_t^\infty}]_i)^2}{[{F_t^m}]_{i,i}}\notag\\
	&+\ln\left(\frac{[{F_t^m}]_{i,i}}{[{F_t^\infty}]_{i,i}}\right)-1\bigg].\label{for:KL}
\end{align}
Finally, simplifying~\cref{for:KL} leads to~\cref{for:KLprop}. 
\end{proof}
\Cref{prop:KL} shows that the error is quantified by the drift of mean~$\Mean(\xkp\vert\vxi_t,\Xi^m_t,\mathcal{D})$ and variance~$\Var(\xkp\vert\vxi_t,\Xi^m_t,\mathcal{D})$ with respect to the true distribution. Therefore, depending on the maximum length of memory~$\overline{m}$, the approximation error is zero at the beginning as the following corollary points out.
\begin{cor}\label{cor:KL}
	For all~$t\leq \overline{m}$, the approximated distribution~$p(\xkp\vert\vxi_t,\Xi^m_t,\mathcal{D})$ given by~\cref{for:gp_ssm_m} equals the true distribution given by~\cref{for:gp_ssm} with KL-divergence~$d_{\mathrm{KL}}(\xkp \| \xkp^m)=0$.
\end{cor}
\begin{proof}
	The corollary is a direct consequence of~\cref{prop:KL}. If the time step~$t$ is equal to or less than the maximum length of memory~$\overline{m}$, the matrices of past states and inputs of the GP-SSM and the GP-ASSM is identical, i.e.,~$\Xi^m_t=\Xi^\infty_t$, and thus, the mean and variance of the approximated distribution equals the true distribution. In consequence, the KL-divergence is zero given by~\cref{for:KL}.
\end{proof}
The restriction of~\cref{prop:KL} that the current state must not be part of the past states is necessary as otherwise, the variance~$F_t(\vxi_t,\Xi^\infty_t)$ or~${F_t^m}$ would be zero. In~\cref{sec3:exam:approx}, this case is explained as the past states and inputs are added to the extended data set such that the predicted variance becomes zero. Additionally, the asymmetry of the KL divergence might be obstructive in some applications. Therefore, we introduce a different measure for the approximation error, namely the mean square prediction error (MSPE).
\begin{prop}\label{prop:mspe}
Consider a GP-ASSM with maximum memory length~$\overline{m}\in\N$ and data set~$\mathcal{D}$ such that
\begin{align}
\xkp^m&\sim \ND\big(\Mean(\xkp\vert\vxi_t,\Xi^m_t,\mathcal{D}),\Var(\xkp\vert\vxi_t,\Xi^m_t,\mathcal{D})\big)\notag
\end{align}
with~$\x_0\in\R^{n_x}$. For given past states and inputs~$\vxi_{0:t}$, the MSPE between~$\xkp^m$ and~$\xkp$ of the GP-SSM is given by
\begin{align}
	\ev{\bigg[\!\Verts{\xkp-\xkp^m}^2\!\bigg]}&\!=\!\Verts{{\dyn_t^\infty}\!-\!{\dyn_t^m}}+\tr\left({F_t^\infty}+{F_t^m}\right).\label{for:mspe}
\end{align}
\end{prop}
\begin{proof}
Since each element of~$\xkp$ and~$\xkp^m$ with a given history of past states and inputs~$\vxi_{0:t}$ is Gaussian distributed, the MSPE is defined by
\begin{align}
	&\ev{\bigg[\Verts{\xkp-\xkp^m}^2\bigg]}=\sum_{i=1}^{n_x}  \ev{\big[(x_{t+1,i}-x_{t+1,i}^m)^2\big]}\label{sec3:for:mspeproof}\\
	&=\sum_{i=1}^{n_x}\big([{f_t^\infty}]_i-[{f_t^m}]_i\big)^2+[{F_t^\infty}]_{i,i}+[{F_t^m}]_{i,i}.\notag
\end{align}
\Cref{sec3:for:mspeproof} is then rewritten to~\cref{for:mspe}.
\end{proof}
With~\cref{prop:mspe,prop:KL} the error of the approximation can be computed. Even if the error measures do not decrease in general for increasing maximum length of memory~$\overline{m}$, the behavior of the variance can be quantified. The next proposition allows to overestimate the predicted variance based on the maximum length of memory.
\begin{prop}\label{prop:var}
	Consider two GP-ASSMs with states and inputs~$\vxi_{0:t}\in\R^{n_\vxi}$ with~$\vxi_{0}\neq\vxi_{1}\neq\ldots\neq\vxi_{t}$ such that $\xkp^m\sim \ND\big(\dyn_t(\vxi_t,\Xi^m_t),F_t(\vxi_t,\Xi^m_t)\big)$ and $\xkp^{{m^\prime}}\sim \ND\big(\dyn_t(\vxi_t,\Xi^{{m^\prime}}_t),F_t(\vxi_t,\Xi^{{m^\prime}}_t)\big)$,
	 where~$\overline{m}$ and~$\overline{m}^\prime$ are the maximum length of memory, respectively. Then, for~$\overline{m}^\prime>\overline{m}$
	\begin{align}
		\tr\left(\Var(\xkp^{{m^\prime}}\vert\vxi_t,\Xi^{{m^\prime}}_t,\mathcal{D})\right)<\tr\left(\Var(\xkp^{m}\vert\vxi_t,\Xi^{m}_t,\mathcal{D})\right)\label{for:trace}
	\end{align}
	holds for all~$t\in\N$ with~$t>\overline{m}$.
\end{prop}
\begin{proof}
Following~\cref{for:gpssm_pred}, the variance for each component of the predicted state of a GP-ASSM is given by
\begin{align}
\var(x_{t+1,i}^{m}\vert\vxi_t,\Xi^{m}_t,\mathcal{D})=\,&k(\vxi_t,\vxi_t)-\bm{k}(\vxi_t, X^m_t)^\top {({K}_t^m)}^{-1}\notag\\
&\bm{k}(\vxi_t, X^m_t).\label{for:varproof}
\end{align}
The Gram matrix~${K}_t^m$ is positive definite and from~\cref{for:gramm} we know, that its dimension is~$({n_\D}+\ubar{m})\times({n_\D}+\ubar{m})$. Based on~${K}_t^m$, the Gram matrix~${K}_t^{m^\prime}\in\R^{({n_\D}+\ubar{m}^\prime)\times({n_\D}+\ubar{m}^\prime)}$ is determined as
\begin{align}
{K}_t^{{m^\prime}}\!\!\!=\!\!\begin{bmatrix}
		\!K(\vxi_{t-\overline{m}^\prime:t-\overline{m}-1},\vxi_{t-\overline{m}^\prime:t-\overline{m}-1})\! & \!\!\!K(\vxi_{t-\overline{m}^\prime:t-\overline{m}-1},X)\!\\
		K(\vxi_{t-\overline{m}^\prime:t-\overline{m}-1},X)^\top & {K}_t^{m}
	\end{bmatrix}\!.\notag
\end{align}
Since the~${K}_t^{m^\prime}$ is positive definite and~$\overline{m}^\prime>\overline{m}$, the inequality
\begin{align}
&k(\vxi_t,\vxi_t)-\bm{k}(\vxi_t, X^{m^\prime}_t)^\top {({K}_t^{m^\prime})}^{-1} \bm{k}(\vxi_t, X^{m^\prime}_t)\notag\\
<&k(\vxi_t,\vxi_t)-\bm{k}(\vxi_t, X^m_t)^\top {({K}_t^m)}^{-1}\bm{k}(\vxi_t, X^m_t)\notag\\
\Rightarrow &\var(x_{t+1,i}^{m^\prime}\vert\vxi_t,\Xi^{m^\prime}_t,\mathcal{D})<\var(x_{t+1,i}^m\vert\vxi_t,\Xi^m_t,\mathcal{D})\label{sec3:for:varineq}
\end{align}
holds for all~$t\in\N$ with~$t>\overline{m}$. Summing up~\cref{sec3:for:varineq} over all elements of~$\x_{t+1}$ leads to~\cref{for:trace}.
\end{proof}
\Cref{prop:var} verifies that the variance of the distribution for the next state ahead~$\xkp^{m^\prime}$ is less than the variance of~$\xkp^m$ with a shorter actual length of memory. This induces that the variance is the lowest for the true sampling as it is given for~$\overline{m}=\infty$. The restriction~$t>m$ in~\cref{prop:var} is necessary as otherwise the variances would be equal for~$t\leq \overline{m}$ as explained in~\cref{cor:KL}. The inequality of past states is necessary to ensure that the GP-ASSM with maximum length of memory~${\overline{m}^\prime}$ contains not only a multiple of the same states which would not decrease the variance. For the sake of completeness, a weaker description for all~$t\in\N$ is provided by the following corollary.
\begin{cor}\label{cor:var}
	Consider two GP-ASSMs with states and inputs~$\vxi_{0:t}\in\R^{n_\vxi}$ such that $\xkp^m\sim \ND\big(\dyn_t(\vxi_t,\Xi^m_t),F_t(\vxi_t,\Xi^m_t)\big)$ and $\xkp^{{m^\prime}}\sim \ND\big(\dyn_t(\vxi_t,\Xi^{{m^\prime}}_t),F_t(\vxi_t,\Xi^{{m^\prime}}_t)\big)$, where the indices~$\overline{m}$ and~$\overline{m}^\prime$ denote the maximum length of memory, respectively. Then, for~$\overline{m}^\prime>\overline{m}$, $\tr\big[\Var(\xkp^{{m^\prime}}\vert\vxi_t,\Xi^{{m^\prime}}_t,\mathcal{D})\big]\leq\tr\big[\Var(\xkp^{m}\vert\vxi_t,\Xi^{m}_t,\mathcal{D})\big]$ holds for all~$t\in\N$.
\end{cor}
\begin{proof}
	The corollary is a direct consequence of~\cref{prop:var} since as long as the current time step~$t$ is less than the maximum length of memory~$\overline{m}$, the variance of~$\xkp^m$ and~$\xkp^{{m^\prime}}$ is identical as shown in~\cref{cor:KL}.
\end{proof}
In the next example, a comparison of the presented error measures and the behavior of the variance is presented.
\begin{exam}
\label{sec3:exam:kl}
In~\cref{fig:kl}, the distributions (gray shaded) for the next state ahead~$x_{t+1}^m$ depending on the maximum length of memory~$\overline{m}$ for a given trajectory~$x_0,\ldots,x_3$ (red circles) is shown. We use here a one-dimensional GP-ASSM with squared exponential function. For sake of simplicity, the input is set to zero, i.e.,~$u_t=0$ for all~$t\in\N$. With increasing maximum length of memory~$\overline{m}$, the variance of the distributions (gray shaded) decreases as stated in~\cref{prop:var}. For~$\overline{m}=3$, the distribution is equal to the true distribution as stated in~\cref{cor:KL}.~\Cref{tab:comp} shows the computed KL-divergence, the MSPE and the variance of~$x_4^m$ per maximum length of memory~$\overline{m}$. 
	\begin{center}
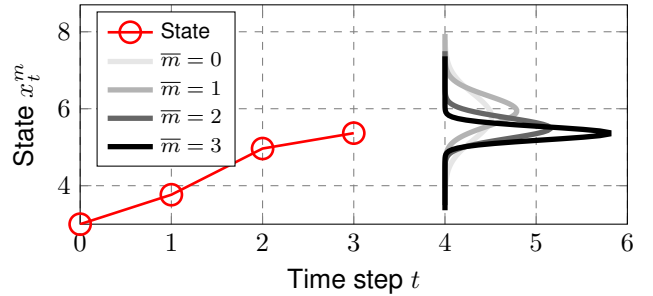

		\tikzsetnextfilename{section3_KL}
		\captionsetup{type=figure}\begin{tikzpicture}
\begin{axis}[
    name=plot1,
  xlabel={Time step $t$},
  ylabel={State $x_t^m$},
  legend pos=north west,
  grid style={dashed,gray},
  grid = both,
  width=\columnwidth,
  height=4.5cm,
  ymin=3,
  ymax=8.7,
  xmin=0,
  xmax=6,
  font={\sffamily},
  legend style={font=\footnotesize\sffamily},
  legend cell align={left},
  legend pos={north west}]
\addplot[color=red,line width=1pt,mark=o,mark size=4] table [x index=0,y index=1]{data/KL_states.dat};
\addplot[color=black!10,line width=2pt] table [x index=0,y index=1]{data/KL_dist.dat};
\addplot[color=black!30,line width=2pt] table [x index=2,y index=3]{data/KL_dist.dat};
\addplot[color=black!60,line width=2pt] table [x index=4,y index=5]{data/KL_dist.dat};
\addplot[color=black!100,line width=2pt] table [x index=6,y index=7]{data/KL_dist.dat};
\legend{State,$\overline{m}=0$,$\overline{m}=1$,$\overline{m}=2$,$\overline{m}=3$}
\end{axis}
\end{tikzpicture} 
		\vspace{-0.2cm}\captionof{figure}{The distribution for the next state ahead~$x_t^m$ depending on the maximum length of memory~$\overline{m}$.}
		\label{fig:kl}
	\end{center}
\begin{center}
\vspace{0.2cm}
\captionsetup{type=table}\begin{tabular}{lllll}
	\toprule
			&	$\overline{m}=0$		&	$\overline{m}=1$	&$\overline{m}=2$ &$\overline{m}=3$\\
			\midrule
	$d_{\mathrm{KL}}(x_4\|x_4^m)$	&	2.1131  &  3.0811  &  0.5559     &    0\\
	$\mspe(x_4,x_4^m)$ & 0.5720   & 0.5190  &  0.1171  &  0.0575\\
	$\Var(x_4^{m}\vert\vxi_3,\Xi^{m}_3,\mathcal{D})$ & 0.3620  &  0.1519  & 0.0706  &  0.0288\\
    \bottomrule
\end{tabular}
\captionof{table}{Comparison of the KL-divergence, MSPE and variance~$\Sigma$ for GP-ASSMs with different lengths of memory.\label{tab:comp}}
\end{center}
\end{exam}
So far, we obtain a method for sampling from the non-Markovian GP-SSM and introduce the approximated GP-ASSM which is a Markov chain of finite order. This approximation allows to use GP-ASSMs like parametric dynamical models since the state dependencies across time are removed. The approximation error is analyzed based on different measures and illustrated in~\cref{sec3:exam:kl}. 
\begin{rem}
This section focuses on the formal development of GP-ASSMs, but the results are also directly applicable to GP-ANOE models. In this case, the proofs are analogously but with the output~$\yk$ as regressor.
\end{rem}
\section{Boundedness of GPDMs}
\label{sec:bound}
After the introduction of GP-SSMs and GP-ASSMs, the models are analyzed in terms of boundedness. Furthermore, the relation of the boundedness properties between the true and the approximated distribution are investigated.
\subsection{GP State Space Models}
\label{sec:gpsmm}
We start with the general introduction of the boundedness of GP-ASSMs for bounded mean functions and kernels.
\begin{thm}\label{thm:1}
     Consider a GP-ASSM~\cref{for:gp_ssm_m} with maximum memory length~$\overline{m}$, bounded mean $\vert m(\x)\vert\leq m_{max}\in\R_{+}$ and kernel function~$k(\x,\x^\prime)\leq k_{max}\in\R_{+}$ for all $\x,\x^\prime\in\R^{n_\xi}$. Then, for every $\xk^m\in\R^{n_x}$, the state~$\xk^m\in\R^{n_x},t\in\N_{+}$ is ultimately \emph{p-bounded} by
    \begin{align}
        \sup_{t\in\N_{+}}\ev \Vert \x_t\Vert^p&\leq n_x\left(\frac{c_2}{2\pi}\right)^{\frac{1}{2}} \!\int\limits_{\R} \vert z^p \vert  \exp\! \left(-\frac{1}{2}\Vert c_1-z \Vert^2 c_2 \right)\! d z\notag\\
        c_1&=m_{max}+n_\D k_{max} \max_i\Vert (K+\sigma_n I)^{-1}Y_{:,i}\Vert\notag
    \end{align}
and $c_2=k_{max}-\frac{k_{max}^2}{k_{max}+\sigma_n^2}$ for all~$p\in\N_+$.
\end{thm}
\begin{proof}
    We start with the computation of the expected value for a one-dimensional GP-SSM, which equals a GP-ASSM with~$\overline{m}=\infty$, as for any other~$\overline{m}$ the number of considered past states is reduced. For this purpose, we first recall the joint probability distribution of a GP-SSM given by $p(\x_{1:t}\vert\inputu_{0:t},\D)=\left|(2 \pi)^{t} \check{K}\right|^{-\frac{1}{2}}\exp \left(\!-\frac{1}{2}\left(\x_{1 : t}-\check{\bm{m}}_{0 : t-1}\right) \check{K}^{-1}\left(\x_{1 : t}-\check{\bm{m}}_{0 : t-1}\right)^\top\!\right)$ with the conditional covariance matrix $\check{K}_t\in\R^{t\times t}$
    \begin{align}
    		\check{K}_t=&K(\vxi_{0:t-1},\vxi_{0:t-1})-K(\vxi_{0:t-1},X)^\top\big(K+\sigma_n^2 I\big)^{-1}\notag\\
    		&K(\vxi_{0:t-1},X). \label{for:proofK}
    	\end{align}
    	The elements of the mean vector~$\check{\bm{m}}_{0 : t-1}\in\R^{1\times t}$ are
    	\begin{align}
    		\check{m}_i=m(\vxi_i)+K(\vxi_i,X)^\top(K+\sigma_n^2 I)^{-1}(Y-\bm{m}(X))\label{for:proofM}
    	\end{align}
    	for all~$i=\{0,\ldots,t-1\}$ with mean vector~$\bm{m}(X)=[m(X_1),\ldots,m(X_{n_\D})]^\top$. Then, the $p$-th absolute expected value is given by
    \begin{align}
        &\sup_{t\in\N_{+}}\ev \vert x_t\vert^p=  \sup_{t\in\N_{+}}\int\limits_{\R^t}  \vert x_t^p\vert p(\x_{1:t}\vert\inputu_{0:t},\D) d \x_{1:t}\label{for:proofest}\\
        &=\sup_{t\in\N_{+}} \int\limits_{\R^t} \vert x_t^p \vert \left|(2 \pi)^{t} \check{K}_t\right|^{-\frac{1}{2}}  \exp\! \left(\!-\frac{1}{2}\Mean_t^{\top} \check{K}_t^{-1}\Mean_t\!\right)\!d \x_{1:t}\notag
    \end{align}
    with~$\Mean_t=\x_{1 : t}-{\bm{m}}_{0 : t-1}$. Note that the mean $\Mean_t$ and the covariance $\check{K}_t$ are depend on the past states and inputs. Thus, the joint distribution is not a multivariate Gaussian distribution such that there exists no analytical solution for \cref{for:proofest} in general. However, we exploit the Gaussian like structure of the distribution to find an upper bound for the integral. First, the matrix~$\check{K}_t$ is positive definite, and its largest eigenvalue $\geig(\check{K}_t)$ is lower bounded by $k_{max}-k_{max}^2/(k_{max}+\sigma_n^2)\leq\geig(\check{K}_t)$ for all~$\vxi_{0:t-1}\in\R^{n_\xi \times t}$ using the Courant-Fischer Theorem. The variable $\sigma_n^2$ is the variance of the noise that corrupts the training data. Second, the elements $\check{m}_i$ of the mean vector are bounded by $\vert \check{m}_i \vert \leq m_{max}+n_\D k_{max}\Vert (K+\sigma_n I)^{-1}Y\Vert$ for bounded mean functions, see~\cite{beckers:ecc2016}. These bounds leads to the upper bound of the expected value given by
    \begin{align}
        &\sup_{t\in\N_{+}}\ev \vert x_t\vert^p\leq \sup_{t\in\N_{+}}\int\limits_{\R^t} \vert x_t^p \vert (2 \pi)^{-\frac{t}{2}}( k_{max}-\frac{k_{max}^2}{k_{max}+\sigma_n^2})^{\frac{t}{2}}\notag\\
        &\exp\! \left(\!-\frac{1}{2}(\bm{c}_t-\x_{1:t})( k_{max}-\frac{k_{max}^2}{k_{max}+\sigma_n^2})(\bm{c}_t-\x_{1:t})^\top\!\right)\!d \x_{1:t}\notag
    \end{align}
    where the elements of the vector $\bm{c}_t\in\R^{1\times t}$ are $c_{t,i}=m_{max}+n_\D k_{max}\Vert (K+\sigma_n I)^{-1}Y\Vert$. Finally, the upper bound can be simplified as the components of the integral are independent such that
    \begin{align}
        \sup_{t\in\N_{+}}\ev \vert x_t\vert^p&\leq\!\left(\frac{c_2}{2\pi}\right)^{\frac{1}{2}} \!\!\int\limits_{\R} \vert z^p \vert  \exp\! \left(\!-\frac{1}{2}\Vert c_1-z \Vert^2 c_2 \!\right)\! d z\label{sec3:for:boundedfinal}
    \end{align}
    with $c_1=m_{max}+n_\D k_{max}\Vert (K+\sigma_n I)^{-1}Y\Vert$ and $c_2=k_{max}-k_{max}^2/(k_{max}+\sigma_n^2)$. As a $n_x$-dimensional GP-SSM depends on separated GPs and the Gram matrix~${K}$ remains bounded, ~\Cref{sec3:for:boundedfinal} can be extended to higher-dimensional~$\xk\in\R^{n_x}$. Consequently, the upper bound in \cref{thm:1} holds. Finally, this remains obviously true for GP-ASSMs with~$\overline{m}<\infty$ as only a subset of past states is considered. This concludes the proof.
\end{proof}
\begin{rem}
    Many commonly used kernels for GPDMs are bounded, for instance, the squared exponential or Mat\'ern kernel.
\end{rem}
As no boundedness of the input~$\inputu_t$ is required for~\cref{thm:1}, we can derive the following corollary for the boundedness of a closed-loop with a GP-ASSM.
\begin{cor}\label{prop:1}
    Consider a GP-ASSM~\cref{for:gp_ssm_m} with maximum memory length~$\overline{m}$, bounded mean $\vert m(\x)\vert\leq m_{max}\in\R_{+}$ and kernel function~$k(\x,\x^\prime)\leq k_{max}\in\R_{+}$ for all $\x,\x^\prime\in\R^{n_\xi}$. A state feedback law is applied such that $\inputu_t=\bm{g}(\x_t)$ with $g\colon\R^{n_x}\to\R^{n_u}$. Then for every $\xk^m\in\R^{n_x}$, the state~$\xk^m\in\R^{n_x},t\in\N_{+}$ of the closed-loop is ultimately \emph{p-bounded} by
    \begin{align}
        \sup_{t\in\N_{+}}\ev \Vert \x_t\Vert^p&\leq n_x\left(\frac{c_2}{2\pi}\right)^{\frac{1}{2}} \!\int\limits_{\R} \vert z^p \vert  \exp\! \left(-\frac{1}{2}\Vert c_1-z \Vert^2 c_2 \right)\! d z\notag\\
        c_1&=m_{max}+n_\D k_{max} \max_i\Vert (K+\sigma_n I)^{-1}Y_{:,i}\Vert\notag
    \end{align}
    and $c_2=k_{max}-\frac{k_{max}^2}{k_{max}+\sigma_n^2}$ for all~$p\in\N_+$.
\end{cor}
\begin{proof}
    The bound for the closed-loop is a direct consequence of~\cref{thm:1} as it holds for arbitrary inputs $\inputu_t\in\R^{n_u}$.
\end{proof}
\Cref{thm:1,prop:1} show the boundedness of GP-ASSMs for bounded mean function and kernel, which holds for the true as well as for the approximated distribution. However, it is also possible that a GP-ASSM with unbounded kernel leads to bounded dynamics. This mainly depends on the training data. In this case, the boundedness property might be lost for a different maximum length of memory, as the following proposition states.
\begin{thm}
\label{prop:notpbounded}
	Consider two GP-ASSMs with the states~$\x^m_t$ and~$\x^{{m^\prime}}_t$, respectively, such that $\xkp^m\sim \ND\big(\dyn_t(\vxi_t,\Xi^m_t),F_t(\vxi_t,\Xi^m_t)\big)$ and $\xkp^{{m^\prime}}\sim \ND\big(\dyn_t(\vxi^\prime_t,\Xi^{{m^\prime}}_t),F_t(\vxi^\prime_t,\Xi^{{m^\prime}}_t)\big)$, where~$\overline{m}$ and~$\overline{m}^\prime$ are the maximum length of memory. Then, for~$\overline{m}<\overline{m}^\prime$
	\begin{align}
		\sup_{t\in\N}\ev\Verts{\xk^m}^{p}<\infty\nRightarrow\sup_{t\in\N}\ev\Vert\xk^{{m^\prime}}\Vert^{p}<\infty\label{sec3:for:proppbound1}
	\end{align}
	holds for any~$p\in\N$ and~$\x^m_0=\x^{{m^\prime}}_0\in\R^{n_x}$.
\end{thm}
\begin{proof}\label{prf:notpbounded}
    We use a counter example to prove this theorem. Consider a one-dimensional GP-ASSM with~$\overline{m}=0$ and linear kernel~$k(\bm{z},\bm{z}^\prime)=\bm{z}^\top \bm{z}^\prime$, where~$\bm{z},\bm{z}^\prime\in\R^n$. We assume two training points at~$X_1=[-1;0],X_2=[1;0]$ and~$Y=[Y_1,Y_2]\in\R^2$ with noise~$\sigma_n^2=1$ and input~$u_t=0$. Using the definition of~\cref{for:gp_ssm_m}, the mean~$f_t$ and variance~$F_t$ of next state~$x_{t+1}^0$ is given by
    \begin{align}
        f_t(\vxi_t,\emptyset)=\frac{1}{3}x_t^0(Y_2-Y_1),\,F_t(\vxi_t,\emptyset)=\frac{1}{3}{(x_t^0)}^2.\label{for:lin_ex}
    \end{align}
    For~$\vert Y_2-Y_1\vert\leq 3$, the sequence~$\{x_t^0\},t\in\N$ is p-bounded, since~$x^0=0$ is stochastically asymptotically stable in the large. Next, in an alternative GP-ASSM, we use the same training points with~${m^\prime}\geq 1$. Starting at~$x_0^{m^\prime}\in\R\backslash 0$, the distribution of~$x_1^{m^\prime}$ can be computed using~\cref{for:lin_ex}. With a Gaussian distributed sampled~$x_1^{m^\prime}$, the next step state~$x_{t+1}^{m^\prime}$ for~$t\geq 1$ are given by
    \begin{align}
    	f_t\left(\begin{bmatrix}x_t^{m^\prime}\\ 0\end{bmatrix},\Xi^{{m^\prime}}_t\right)&=\frac{x_1^{m^\prime}}{x_0^{m^\prime}}x_t^{m^\prime},F_t\left(\begin{bmatrix}x_t^{m^\prime}\\ 0\end{bmatrix},\Xi^{{m^\prime}}_t\right)=0\notag\\
    	x_{t+1}^{m^\prime}&=\frac{x_1^{m^\prime}}{x_0^{m^\prime}}x_t^{m^\prime}.\label{for:lin_sys}
    \end{align}
    The predicted variance for all states in the future is zero, since the state~$x_1^{m^\prime}$ exactly defines a sample of the GP with a linear kernel. The reason is that a linear function is fully defined by one point unequal zero. Based on the Gaussian distribution of~$x_1^{m^\prime}$, the probability, that a trajectory of~\cref{for:lin_sys} is unbounded, is computed by
    \begin{align}
    	\Prob\left(\vert x_1^{m^\prime}/x_0^{m^\prime}\vert\!>\!1\right)&=1+\operatorname{cdf}\left[(-3\vert x_0^{m^\prime}\vert+\Delta Y)/([x_0^{m^\prime}]^2)\right]\notag\\
    	&-\operatorname{cdf}\left[(3\vert x_0^{m^\prime}\vert+\Delta Y)/([x_0^{m^\prime}]^2)\right],\label{sec3:for:propproofbound}
    \end{align}
    where~$\Delta Y=Y_1-Y_2$ and~$\operatorname{cdf}$ denotes the standard normal cumulative distribution function. Since the probability~\cref{sec3:for:propproofbound} is greater than zero, the sequence~$\{x_t^{m^\prime}\},t\in\N$ is not p-bounded. Hence, a different maximum length of memory~$\overline{m}$ of a GP-ASSM might lead to a loss boundedness property as stated in~\cref{prop:notpbounded}.
\end{proof}
\begin{exam}
In~\cref{fig:CL}, the counter example from the proof of~\cref{prop:notpbounded} is visualized. For this purpose, we employ two GP-ASSMs with $\overline{m}=0$ and $\overline{m}=10$, respectively, based on a linear kernel~$k(\bm{z},\bm{z}^\prime)=\bm{z}^\top \bm{z}^\prime$. Although the samples of the GP-ASSM with~$\overline{m}=0$ are bounded (top), a GP-ASSM with~$\overline{m}^\prime=10$ (bottom) shows unbounded trajectories, which leads to an unbounded mean and variance.
\begin{figure}[tbh]
	\begin{center}
		\tikzsetnextfilename{lin_system}
		\captionsetup{type=figure}\begin{tikzpicture}
\begin{axis}[
name=plot1,
  ylabel={$x_t^0$},
  legend pos=north west,
  grid style={dashed,gray},
  grid = both,
       width=\columnwidth,
  height=4cm,
  ymin=-1,
  ymax=11,
  xmin=0,
  xmax=25,
  font={\sffamily},
  legend style={font=\footnotesize\sffamily},
  legend cell align={left},
  legend pos={north east},
  xticklabels={,,}]
  \addplot[color=blue,line width=2pt] table [x index=0,y index=1]{data/lin_system_2.dat};
  \addplot+[name path=varp1, color=gray,opacity=0.3, no marks] table [x index=0,y expr=\thisrowno{1}+\thisrowno{2}]{data/lin_system_2.dat};
\addplot+[name path=varm1, color=gray,opacity=0.3, no marks] table [x index=0,y expr=\thisrowno{1}-\thisrowno{2}]{data/lin_system_2.dat};
\addplot[gray,opacity=0.5] fill between[ of = varm1 and varp1]; 
\addplot[color=Goldenrod,line width=1pt,mark=+,mark size=4] table [x index=0,y index=3]{data/lin_system_2.dat};
\addplot[color=NavyBlue,line width=1pt,mark=+,mark size=4] table [x index=0,y index=4]{data/lin_system_2.dat};
\addplot[color=Fuchsia,line width=1pt,mark=+,mark size=4] table [x index=0,y index=5]{data/lin_system_2.dat};
\addplot[color=Emerald,line width=1pt,mark=+,mark size=4] table [x index=0,y index=6]{data/lin_system_2.dat};
\addplot[color=RawSienna,line width=1pt,mark=+,mark size=4] table [x index=0,y index=7]{data/lin_system_2.dat};
\addplot[only marks,color=black,line width=1pt,mark=+,mark size=4] coordinates {(-10,-10)};
\legend{Mean,,,Standard deviation,,,,,,Samples}
\end{axis}
\begin{axis}[
    name=plot2,
    at=(plot1.below south west), anchor=above north west,
  xlabel={Time step $t$},
  ylabel={$x_t^{10}$},
  legend pos=north west,
  grid style={dashed,gray},
  grid = both,
       width=\columnwidth,
  height=3.8cm,
  ymin=-1,
  ymax=40,
  xmin=0,
  xmax=25,
  font={\sffamily},
  legend style={font=\footnotesize\sffamily},
  legend cell align={left}]
\addplot+[name path=varp1, color=gray,opacity=0.3, no marks] table [x index=0,y expr=\thisrowno{1}+\thisrowno{2}]{data/lin_system_1.dat};
\addplot+[name path=varm1, color=gray,opacity=0.3, no marks] table [x index=0,y expr=\thisrowno{1}-\thisrowno{2}]{data/lin_system_1.dat};
\addplot[gray,opacity=0.5] fill between[ of = varm1 and varp1]; 
\addplot[color=Goldenrod,line width=1pt,mark=+,mark size=4] table [x index=0,y index=3]{data/lin_system_1.dat};
\addplot[color=NavyBlue,line width=1pt,mark=+,mark size=4] table [x index=0,y index=4]{data/lin_system_1.dat};
\addplot[color=Fuchsia,line width=1pt,mark=+,mark size=4] table [x index=0,y index=5]{data/lin_system_1.dat};
\addplot[color=Emerald,line width=1pt,mark=+,mark size=4] table [x index=0,y index=6]{data/lin_system_1.dat};
\addplot[color=RawSienna,line width=1pt,mark=+,mark size=4] table [x index=0,y index=7]{data/lin_system_1.dat};
\addplot[color=blue,line width=2pt] table [x index=0,y index=1]{data/lin_system_1.dat};
\end{axis}
\end{tikzpicture} 
		\vspace{-0.4cm}
		\captionof{figure}{The GP-ASSM with~$\overline{m}=0$ (top) results in bounded system trajectories whereas a GP-ASSM with~$\overline{m}^\prime\geq 1$ (bottom) generates unbounded trajectories. Therefore the boundedness property is lost for different maximum lengths of memory~$\overline{m}$.}
		\vspace{-0.4cm}
		\label{fig:CL}
	\end{center}
	\end{figure}
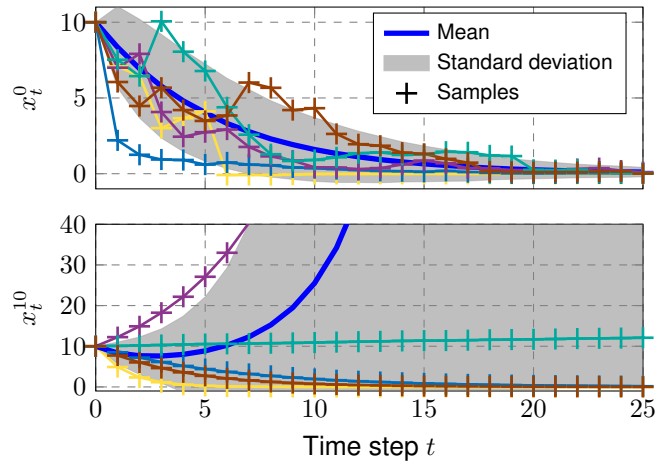
\end{exam}
The following theorem shows the relationship between the boundedeness of GP-ASSMs with different length of memory. It states that the approximated dynamics given by a GP-ASSM is bounded if the dynamics of the GP-SSM is bounded. Thus, it allows to use the approximation in control settings without losing the boundedness, which is important for the robustness and stability analysis. Note, that in contrast to~\cref{thm:1}, the kernel is not required to be bounded.
\begin{thm}
\label{prop:pbounded}
	Considering two GP-ASSMs with the states~$\x^m_t$ and~$\x^{{m^\prime}}_t$, respectively, such that $\xkp^m\sim \ND\big(\dyn_t(\vxi_t,\Xi^m_t),F_t(\vxi_t,\Xi^m_t)\big)$ and $\xkp^{{m^\prime}}\sim \ND\big(\dyn_t(\vxi^\prime_t,\Xi^{{m^\prime}}_t),F_t(\vxi^\prime_t,\Xi^{{m^\prime}}_t)\big)$, where~$\overline{m}$ and~$\overline{m}^\prime$ are the maximum length of memory. Then, if~$\overline{m}<\overline{m}^\prime$,
	\begin{align}
		\sup_{t\in\N,\x^{{m^\prime}}_0\in\R^{n_x}}\!\!\!\!\ev\Vert\xk^{{m^\prime}}\Vert^{p}<\infty\Rightarrow\!\!\!\sup_{t\in\N,\x^m_0\in\R^{n_x}}\!\!\!\!\ev\Verts{\xk^m}^{p}<\infty\label{sec3:for:proppbound2}
	\end{align}
	holds for all~$p\in\N$.
\end{thm}
\begin{rem}
Note the swap of~$\xk^{m^\prime}$ and~$\xk^{m}$ in~\cref{sec3:for:proppbound2} in contrast to~\cref{sec3:for:proppbound1}.
\end{rem}
\begin{proof}
    In the following, we split the proof in two parts depending on time step~$t$.\\
    For~$t\leq \overline{m}$, the memories~$\Xi^m_t$ and~$\Xi^{{m^\prime}}_t$ of both GP-ASSMs are identical and, thus, the expected value is bounded by $\sup_{t\in\N,t\leq \overline{m}}\ev[{(\x_t^{{m^\prime}})}^p]=\sup_{t\in\N,t\leq m}\ev[{(\x_t^{m})}^p]<\infty$. For~$t>m$, we use the last point in memory~$\x^{{m^\prime}}_{\max{(0,t-{m^\prime}-1)}}$ as initial point for~$\xkp^m$. Thus, we can follow the above argumentation again, which leads to  $\sup_{t\in\N,t>\overline{m}}\ev[{(\x_t^{m})}^p]<\infty$ such that the boundedness is preserved.
\end{proof}
\subsection{GP Nonlinear Output Error Models}
In this section, we transfer our results about boundedness of GP-ASSMs to GP-ANOE models. In GP-ANOE models, the feedback loop is closed by the output~$\yk$ instead of the state~$\xk$ as in GP-ASSMs. Therefore, we present the following results without further explanation and refer here to~\cref{sec:gpsmm}.
\begin{prop}\label{thm:1noe}
    Consider a GP-ANOE~\cref{for:gp_noe_m} with maximum memory length~$\overline{m}$ and bounded mean $\vert m(\x)\vert\leq m_{max}\in\R_{+}$ and kernel function~$k(\x,\x^\prime)\leq k_{max}\in\R_{+}$ for all $\x,\x^\prime\in\R^{n_\xi}$. Then for every $\y_0\in\R^{n_y}$, the output~$\yk^m\in\R^{n_y},t\in\N_{+}$ is ultimately \emph{p-bounded} by
    \begin{align}
        \sup_{t\in\N_{+}}\ev \Vert \y_t\Vert^p&\leq n_y\left(\frac{c_2}{2\pi}\right)^{\frac{1}{2}} \!\int\limits_{\R} \vert y^p \vert  \exp\! \left(-\frac{1}{2}\Vert c_1-x \Vert^2 c_2 \right)\! d x\notag\\
        c_1&=m_{max}+n_\D k_{max} \max_i\Vert (K+\sigma_n I)^{-1}Y_{:,i}\Vert\notag
    \end{align}
and  $c_2=k_{max}-\frac{k_{max}^2}{k_{max}+\sigma_n^2}$ for all~$p\in\N_+$.
\end{prop}
\begin{proof}
    Analogously to the proof of~\cref{thm:1} with the GP-ANOE model defined by~\cref{for:gp_noe_m}.
\end{proof}
\begin{prop}
	Consider two GP-ANOEs with outputs~$\y^m_t$ and~$\y^{{m^\prime}}_t$, respectively, such that the output $\ykp^m\sim \ND\big(\bm{h}_t(\z^m_t,\Lambda_t),H_t(\z^m_t,\Lambda_t)\big)$ and $\ykp^{{m^\prime}}\sim \ND\big(\bm{h}_t(\z^{{m^\prime}}_t,\Lambda^\prime_t),H_t(\z^{{m^\prime}}_t,\Lambda^\prime_t)\big)$
	 where~$\overline{m}$ and~$\overline{m}^\prime$ are the maximum length of memory, respectively. Then, for~$\overline{m}<\overline{m}^\prime$, $\sup_{t\in\N}\ev\Verts{\ykp^m}^{p}<\infty\nRightarrow\sup_{t\in\N}\ev\Vert\ykp^{{m^\prime}}\Vert^{p}<\infty$
	holds for~$p\in\N$ and~$\z^m_0=\z^{{m^\prime}}_0\in\R^{n_\zeta}$.
	\label{prop:notpboundednoe}
\end{prop}
\begin{proof}
Analogously to the proof of~\cref{prop:notpbounded} with the GP-ANOE model defined by~\cref{for:gp_noe_m}.
\end{proof}
\begin{prop}
\label{prop:pboundednoe}
	Consider two GP-ANOE models with outputs~$\y^m_t$ and~$\y^{{m^\prime}}_t$, respectively, such that $\ykp^m\sim \ND\big(\bm{h}_t(\z^m_t,\Lambda_t),H_t(\z^m_t,\Lambda_t)\big)$ and $\ykp^{{m^\prime}}\sim \ND\big(\bm{h}_t(\z^{{m^\prime}}_t,\Lambda^\prime_t),H_t(\z^{{m^\prime}}_t,\Lambda^\prime_t)\big)$, where~$\overline{m}$ and~$\overline{m}^\prime$ are the maximum length of memory, respectively. Then, if~$\overline{m}<\overline{m}^\prime$ holds, $\sup_{t\in\N,\z^{{m^\prime}}_0\in\R^{n_y}}\ev\Vert \yk^{{m^\prime}}\Vert^{p}<\infty\Rightarrow\sup_{t\in\N,\z^m_0\in\R^{n_x}}\ev\Verts{\yk^m}^{p}<\infty$ holds for all~$p\in\N$.
\end{prop}
\begin{proof}
Analogously to the proof of~\cref{prop:pbounded} with the GP-ANOE model defined by~\cref{for:gp_noe_m}.
\end{proof}
\section{Case study}
In two case studies, we demonstrate the modeling with GP-ASSMs and discuss their behavior.
\subsection{Open-loop}
In an open-loop setting, we show the modeling of a dynamical system with a GP-SSM and GP-ASSMs with different maximum lengths of memory. As dynamical system to be modeled, we consider the non-autonomous discrete-time predator–prey system introduced in~\cite{liu2010note}. It is given by
\begin{align}
    \begin{split}
        x_{t+1,1}&=x_{t,1} \exp\left(1-0.4 x_{t,1}-\frac{(2+1.2 u_{t,1}) x_{t,2}}{1+(x_{t,1})^2} \right)\\
        x_{t+1,2}&=x_{t,2} \exp\left(1+0.5 u_{t,1}-\frac{(1.5- u_{t,2}) x_{t,2}}{x_{t,1}} \right)
    \end{split}\label{for:dynpp}
    \end{align}
    with input and noisy output
    \begin{align}
    \yk&=\xk+\bm{\nu},\quad
    \uk=\begin{bmatrix}
    \cos(0.02\pi t)\\ \sin(0.02\pi t)
    \end{bmatrix},\label{for:inputoutputpp}
\end{align}
and with state~$\x_{t}\in\R^2$, output~$\yk\in\R^2$, input~$\uk\in\R^2$, and Gaussian distributed noise~$\bm{\nu}\in\R^2,\bm{\nu}\sim\mathcal{N}(\bm{0},0.05^2 I)$. The states~$x_{t,1}$ and~$x_{t,2}$ represent the population size of prays and predators, respectively, but are taken to be continuous. The system dynamics~\cref{for:dynpp} are assumed to be unknown whereas the input and output, given by~\cref{for:inputoutputpp}, are assumed to be known. For the modeling with a GP-SSM, 33 training points of a trajectory from the predator–prey system with initial state~$\x_0=[0.3;0.8]$ are collected. More detailed, every third state~$\x_{t}$, input~$\uk$ and output~$\yk$ between~$t=1,\ldots,100$ is recorded. Thus, the training set~$\D=\{X,Y\}$ consists of
\begin{align}
    \begin{split}
        X&=[\vxi_1,\vxi_4,\ldots,\vxi_{97}]\text{ with }\vxi_t=[\x_{t};\uk]\\
        Y&=[\y_1,\y_4,\ldots,\y_{97}]^\top.
    \end{split}\label{for:studydataset}
\end{align}
Following the structure of GP-SSMs in~\cref{for:gp_ssm}, two GPs are employed to model each element of the state~$\xk$ separately. Both GPs are based on a squared exponential kernel with automatic relevance detection given by~$k(\vxi_t,\vxi_t^\prime)=\varphi_1^2 \exp{\left(-(\vxi_t-\vxi_t^\prime)^\top P^{-1}(\vxi_t-\vxi_t^\prime) \right) }$ with matrix~$P=\diag(\varphi_2^2,\ldots,\varphi_5^2)$. This kernel is bounded with respect to~$\vxi_t,\vxi_t^\prime\in\R^4$. The hyperparameters~$\varphi_1,\ldots,\varphi_5$ of each GP are optimized by means of the likelihood function, see~\cite{rasmussen2006gaussian}. In this study, we model the dynamics~\cref{for:dynpp} with a GP-SSM, a GP-ASSM with maximum length of memory~$10$ and a GP-ASSM with maximum length of memory~$0$. For the testing of these models, we select the initial state~$\x_0=[0.268;0.400]$. The top plot of~\cref{fig:case_study_1} visualizes the trajectory of the predator–prey system~\cref{for:dynpp}, considered as the ground-truth. 
\begin{figure}[ht]
	\begin{center}
		\tikzsetnextfilename{case_study_1}
		\vspace{0.15cm}
		\begin{tikzpicture}
\begin{axis}[
name=plot1,
  ylabel={Actual state $\xk$},
  legend pos=north west,
  width=\columnwidth,
  height=4.6cm,
  ymin=0,
  ymax=2.9,
  xmin=0,
  xmax=200,
    font={\sffamily},
  legend style={font=\footnotesize\sffamily,at={(0.53,1)},anchor=north},
  legend cell align={left},
  xticklabels={,,}]
\addplot[color=red,dashed,line width=1pt] table [x index=0,y index=1]{data/case_study_real.dat};
\addplot[color=blue, line width=1pt] table [x index=0,y index=2]{data/case_study_real.dat};
\legend{Prey,Predator}
\end{axis}
\begin{axis}[
    name=plot2,
    at=(plot1.below south west), anchor=above north west,
  ylabel={GP-SSM state $\xk^\infty$},
  legend pos=north west,
       width=\columnwidth,
  height=4.6cm,
  ymin=0,
  ymax=2.9,
  xmin=0,
  xmax=200,
  font={\sffamily},
  legend style={font=\footnotesize\sffamily},
  legend cell align={left},
  xticklabels={,,}]
\addplot[color=red,dashed,line width=1pt] table [x index=0,y index=1]{data/case_study_minf.dat};
\addplot[color=orange,dashed, line width=1pt] table [x index=0,y index=2]{data/case_study_minf.dat};
\addplot[color=yellow,dashed, line width=1pt] table [x index=0,y index=3]{data/case_study_minf.dat};
\addplot[color=blue, line width=1pt] table [x index=0,y index=4]{data/case_study_minf.dat};
\addplot[color=purple, line width=1pt] table [x index=0,y index=5]{data/case_study_minf.dat};
\addplot[color=black!20!blue, line width=1pt] table [x index=0,y index=6]{data/case_study_minf.dat};
\draw (axis cs:35,0.15) rectangle (axis cs:85,2.35);
\draw (axis cs:135,0.15) rectangle (axis cs:185,2.35);
\node[anchor=center] at (axis cs:110,2.1) {identical};
\end{axis}
\begin{axis}[
    name=plot3,
    at=(plot2.below south west), anchor=above north west,
  ylabel={GP-ASSM state $\xk^{10}$},
  legend pos=north west,
       width=\columnwidth,
  height=4.6cm,
  ymin=0,
  ymax=2.9,
  xmin=0,
  xmax=200,
  font={\sffamily},
  legend style={font=\footnotesize\sffamily},
  legend cell align={left},
  xticklabels={,,}]
\addplot[color=red,dashed,line width=1pt] table [x index=0,y index=1]{data/case_study_m10.dat};
\addplot[color=orange,dashed, line width=1pt] table [x index=0,y index=2]{data/case_study_m10.dat};
\addplot[color=yellow,dashed, line width=1pt] table [x index=0,y index=3]{data/case_study_m10.dat};
\addplot[color=blue, line width=1pt] table [x index=0,y index=4]{data/case_study_m10.dat};
\addplot[color=purple, line width=1pt] table [x index=0,y index=5]{data/case_study_m10.dat};
\addplot[color=black!20!blue, line width=1pt] table [x index=0,y index=6]{data/case_study_m10.dat};
\draw (axis cs:35,0.15) rectangle (axis cs:85,2.35);
\draw (axis cs:135,0.15) rectangle (axis cs:185,2.35);
\node[anchor=center] at (axis cs:110,2.1) {different};
\end{axis}
\begin{axis}[
    name=plot4,
    at=(plot3.below south west), anchor=above north west,
  xlabel={Time step $t$},
  ylabel={GP-ASSM state $\xk^0$},
  legend pos=north west,
       width=\columnwidth,
  height=4.6cm,
  ymin=0,
  ymax=2.9,
  xmin=0,
  xmax=200,
  font={\sffamily},
  legend style={font=\footnotesize\sffamily,row sep=-2pt,column sep=0.3cm,at={(0.5,1)},anchor=north},
  legend columns=3,
  transpose legend,
  legend cell align={left}]
\addplot[color=red,dashed,line width=1pt] table [x index=0,y index=1]{data/case_study_m0.dat};\addlegendentry{\raisebox{-2.4mm}[0mm][0mm]{Samples Prey}};
\addplot[color=orange,dashed, line width=1pt] table [x index=0,y index=2]{data/case_study_m0.dat};\addlegendentry{\raisebox{0mm}[0mm][0mm]{ }};
\addplot[color=yellow,dashed, line width=1pt] table [x index=0,y index=3]{data/case_study_m0.dat};\addlegendentry{\raisebox{0mm}[0mm][0mm]{ }};
\addplot[color=blue, line width=1pt] table [x index=0,y index=4]{data/case_study_m0.dat};\addlegendentry{\raisebox{-2.4mm}[0mm][0mm]{Samples Predator}};
\addplot[color=purple, line width=1pt] table [x index=0,y index=5]{data/case_study_m0.dat};\addlegendentry{\raisebox{0mm}[0mm][0mm]{ }};
\addplot[color=black!10!blue, line width=1pt] table [x index=0,y index=6]{data/case_study_m0.dat};\addlegendentry{\raisebox{0mm}[0mm][0mm]{ }};
\end{axis}
\end{tikzpicture} 
		\vspace{-0.3cm}\caption{From top to bottom: Trajectory of predator–prey system, samples of GP-SSM, samples of GP-ASSM with~$\overline{m}=10$, and samples of GP-SSM with~$\overline{m}=0$. For decreasing maximum length of memory of the approximations, the variance is increasing which leads to rougher trajectories.}\vspace{-0.5cm}
		\label{fig:case_study_1}
	\end{center}
\end{figure}
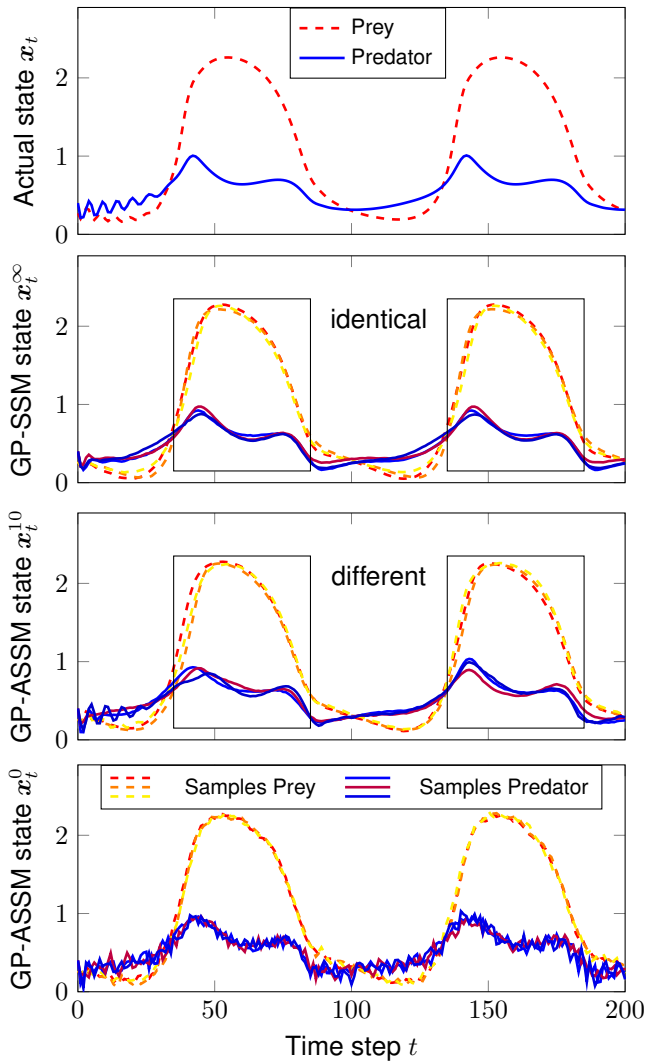
After a transition phase, the numbers of prays (red dashed) and predators (blue solid) converge to a periodic solution. The second plot shows three samples of the GP-SSM drawn by means of~\cref{propy:gpssm}. Even though the training set consists only of data up to the time step~$t=97$, see~\cref{for:studydataset}, the GP-SSM precisely predicts the trajectory after the transition phase. As the GP-SSM implies~$\overline{m}=\infty$, all past state transitions are added to the memory~$\Xi^\infty_t$, defined in~\cref{for:memgpssm}, and used for the next state ahead prediction. Consequently, the shape of each sample is identical in periodic repetitions, as highlighted inside the boxes in the second plot of~\cref{fig:case_study_1}. Three samples of the GP-ASSM with maximum length of memory~$10$, given by the means of~\cref{for:gp_ssm_m}, are visualized in the third plot of~\cref{fig:case_study_1}. The samples are similar to the samples of the GP-SSM, since the memory~$\Xi^{10}_t$ consists of sufficiently many past states to generate a similar predictive distribution for next step state. However, the shape of the samples differs between the periodic repetitions, as indicated with the two boxes. This variation is due to the reduced memory, which induces that the evolution of the state inside the left box is not considered for the prediction of the corresponding state in the right box. In contrast to the GP-SSM, the maximum length of memory~$10$ bounds the size of the Gram matrix~${K}_t^{10}$. In the bottom plot of~\cref{fig:case_study_1}, three samples of the GP-ASSM with maximum length of memory~$0$ are drawn. The variance for each prediction step is significantly higher, as described in~\cref{prop:var}, such that the trajectories are rougher. However, the size of the Gram matrix~${K}_t^{0}$ remains constantly low.\\
Finally, the GP-SSM and the GP-ASSM with~$\overline{m}=0$ are tested with 50 different initial values, which are drawn from a uniform distribution between~$[-5,5]$ for both states, visualized in~\cref{fig:case_study_2}. All trajectories are bounded, which supports~\cref{prop:pbounded,thm:1}.
\begin{figure}[bht]
	\begin{center}
		\tikzsetnextfilename{case_study_2}
		\vspace{0.15cm}
		\begin{tikzpicture}
\begin{axis}[
name=plot1,
  ylabel={State $\xk^\infty$},
  legend pos=north west,
       width=\columnwidth,
  height=4.5cm,
  ymin=-5,
  ymax=5,
  xmin=0,
  xmax=60,
    font={\sffamily},
  legend style={font=\footnotesize\sffamily},
  legend cell align={left}]

\foreach \N in {1,...,50}{
\addplot[red,dashed,mark=none] table [x index=0,y index=\N]{data/case_study_minf_multi.dat};
}
\foreach \N in {51,...,100}{
\addplot[blue,solid,mark=none] table [x index=0,y index=\N]{data/case_study_minf_multi.dat};
}
\end{axis}
\begin{axis}[
name=plot2,
    at=(plot1.below south west), anchor=above north west,
  xlabel={Time step $t$},
  ylabel={State $\xk^0$},
  legend pos=north west,
       width=\columnwidth,
  height=4.5cm,
  ymin=-5,
  ymax=5,
  xmin=0,
  xmax=60,
    font={\sffamily},
  legend style={font=\footnotesize\sffamily},
    legend pos=south east,
  legend cell align={left}]

\addplot[red, dashed,mark=none] table [x index=0,y index=1]{data/case_study_m0_multi.dat};
\addlegendentry{Prey}
\foreach \N in {2,...,50}{
\addplot[red, dashed,mark=none,forget plot] table [x index=0,y index=\N]{data/case_study_m0_multi.dat};
}
\addplot[blue, solid,mark=none] table [x index=0,y index=51]{data/case_study_m0_multi.dat};
\addlegendentry{Predator}
\foreach \N in {52,...,100}{
\addplot[blue, solid,mark=none,forget plot] table [x index=0,y index=\N]{data/case_study_m0_multi.dat};
}
\end{axis}
\end{tikzpicture} 
		\vspace{-0.2cm}\caption{Trajectories of 50 samples starting from multiple initial points demonstrate the boundedness of the GP-SSM and GP-ASSM.}\vspace{-0.5cm}
		\label{fig:case_study_2}
	\end{center}
\end{figure}
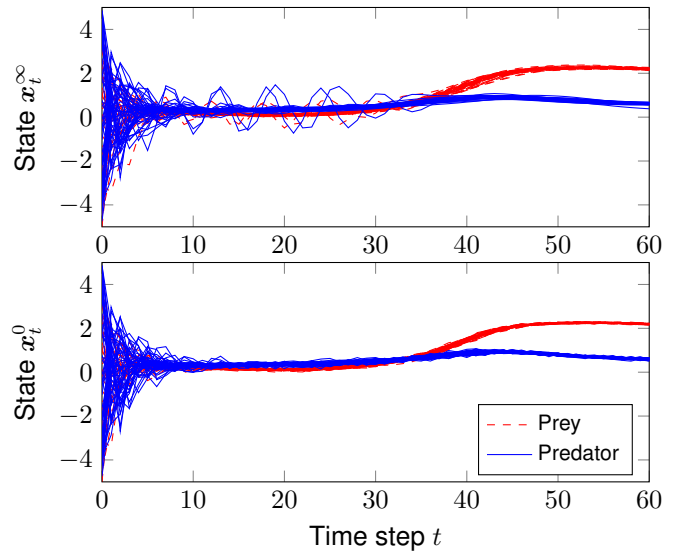
\subsection{Closed-loop}
In the case study, we demonstrate the usage of a GP-ASSM to test a controller for a chaotic dynamical systems. For this purpose, we consider the time-continuous Thomas' cyclically symmetric attractor with an external input described by
\begin{align}
    \dot{\x}&=\begin{bmatrix}
        \sin(x_2)-b x_1 \\ \sin(x_3)-b x_2 \\ \sin(x_1)-b x_3
    \end{bmatrix}+\begin{bmatrix}
        u \\ 0 \\ 0
    \end{bmatrix}\label{for:case2dyn},\,\bm{y}=\x+\bm{\nu}
\end{align}
with state $\x\in\R^3$, input $u\in\R$, output $\bm{y}\in\R^3$ and noise $\bm{\nu}\sim\mathcal{N}(\bm{0},0.006^2 I)$. The constant $b$ is set to $b=0.2$. The resulting trajectories can be seen as the motion of a frictionally dampened particle moving in a 3D lattice of force, see~\cite{sprott2007labyrinth}. The goal is to test the performance of a set-point controller for the dynamics \cref{for:case2dyn} which are assumed to be unknown and costly to evaluate or safety critical. Therefore, a simulation with a GP-ASSM should be performed to evaluate the controller before it is applied to the real system. The training set $\D$ consists of 375 training points equally distributed on the set $[-1,1]^3$ for the state $x$ and $[-2,2]$ for the input $u$. The sample time is set to $\SI{0.01}{\second}$ for a low discretization error as, otherwise, it can lead to a significantly different behavior of the chaotic system. The GP-ASSM with maximum memory length of one is based on squared exponential kernels and the hyperparameters are optimized by means of the likelihood function. The control law for testing is assumed to be $u=-[2,2,2]\x$. The top graph of~\cref{fig:case_study2} visualizes the resulting 20 samples with the mean (solid line) and the $3\sigma$ standard deviation (shaded area). The samples converge to a small neighborhood around zero using the feedback control law. Then, the control law is applied to the actual, time-continuous system where the trajectory converges to zero. The example shows that the GP-ASSM is sufficient to mimic the behavior of an actual system. The benefit in contrast to standard GP-SSMs is the significantly reduced computation time.~\Cref{fig:case_study2_time} shows the computation time\footnote{Simulations were performed on a Intel i7-4600U with 2.1 Ghz, 8 GB RAM, and Matlab 2018.} which is required per time step (top) and the total time (bottom) over the time steps. The GP-SSM computation time for a step scales cubic with the number of time steps due to the required matrix inversion, see~\cref{for:gp_meanvar}. In contrast, the GP-ASSM is constant such that the total time is reduced from $\SI{361}{\second}$ to $\SI{52}{\second}$.
\begin{figure}
	\begin{center}
		\tikzsetnextfilename{case_study2}
		\begin{tikzpicture}
\begin{axis}[
    name=plot1,
  ylabel={GP-ASSM state $\xk^1$},
  legend pos=north east,
       width=0.96\columnwidth,
  height=3.7cm,
  ymin=-0.8,
  ymax=1.2,
  xmin=0,
  xmax=10,
  font={\sffamily},
  legend style={font=\footnotesize\sffamily},
  legend cell align={left},
  xticklabels={,,}]
\addplot+[name path=varp1, color=red,opacity=0.3, no marks] table [x index=0,y expr=\thisrowno{1}+\thisrowno{4}]{data/case_study2_m1_mean_var.dat};
\addplot+[name path=varm1, color=red,opacity=0.3, no marks] table [x index=0,y expr=\thisrowno{1}-\thisrowno{4}]{data/case_study2_m1_mean_var.dat};
\addplot[red,opacity=0.5] fill between[ of = varm1 and varp1]; 
\addplot+[name path=varp2, color=blue,opacity=0.3, no marks] table [x index=0,y expr=\thisrowno{2}+\thisrowno{5}]{data/case_study2_m1_mean_var.dat};
\addplot+[name path=varm2, color=blue,opacity=0.3, no marks] table [x index=0,y expr=\thisrowno{2}-\thisrowno{5}]{data/case_study2_m1_mean_var.dat};
\addplot[blue,opacity=0.5] fill between[ of = varm2 and varp2]; 
\addplot+[name path=varp3, color=black,opacity=0.3, no marks] table [x index=0,y expr=\thisrowno{3}+\thisrowno{6}]{data/case_study2_m1_mean_var.dat};
\addplot+[name path=varm3, color=black,opacity=0.3, no marks] table [x index=0,y expr=\thisrowno{3}-\thisrowno{6}]{data/case_study2_m1_mean_var.dat};
\addplot[black,opacity=0.5] fill between[ of = varm3 and varp3]; 
\addplot[color=red,line width=1pt] table [x index=0,y index=1]{data/case_study2_m1_mean_var.dat};
\addplot[color=blue,line width=1pt] table [x index=0,y index=2]{data/case_study2_m1_mean_var.dat};
\addplot[color=black,line width=1pt] table [x index=0,y index=3]{data/case_study2_m1_mean_var.dat};
\legend{,,,,,,,,,$x_1$,$x_2$,$x_3$}
\end{axis}
\begin{axis}[
    name=plot2,
    at=(plot1.below south west), anchor=above north west,
      xlabel={Time [s]},
  ylabel={Actual state $\x$},
  legend pos=north west,
  width=0.96\columnwidth,
  height=3.7cm,
  ymin=-0.8,
  ymax=1.2,
  xmin=0,
  xmax=10,
    font={\sffamily},
  legend style={font=\footnotesize\sffamily,at={(0.53,1)},anchor=north},
  legend cell align={left}]
\addplot[color=red,line width=1pt] table [x index=0,y index=1]{data/case_study2_real.dat};
\addplot[color=blue,line width=1pt] table [x index=0,y index=2]{data/case_study2_real.dat};
\addplot[color=black,line width=1pt] table [x index=0,y index=3]{data/case_study2_real.dat};
\end{axis}
\end{tikzpicture} 
		\vspace{-0.5cm}\caption{Top: Mean and variance of the GP-ASSM's samples testing the control law. Bottom: Behavior of the actual system with the tested control law.}\vspace{-0.5cm}
		\label{fig:case_study2}
	\end{center}
\end{figure}
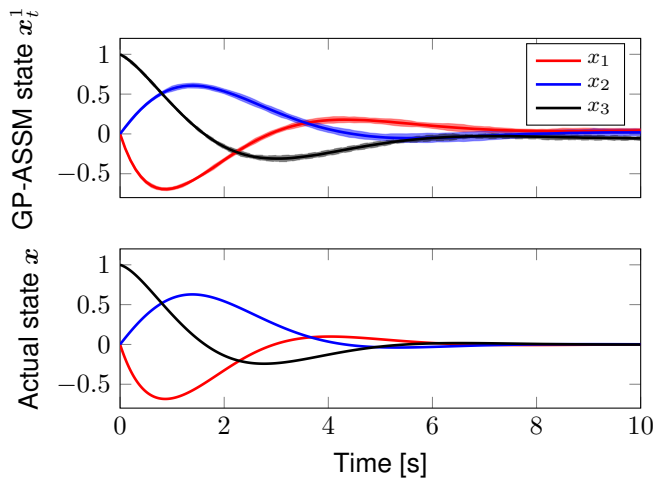
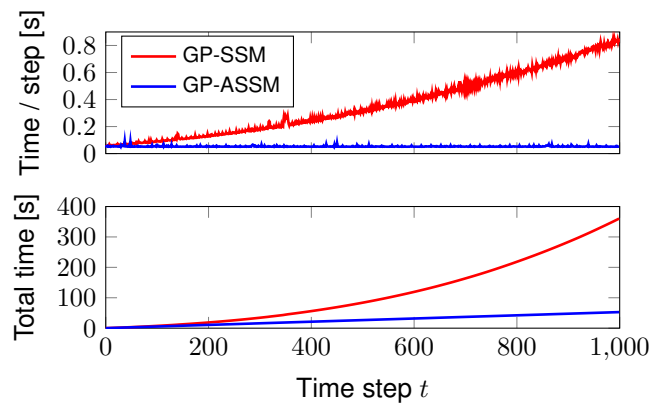
\begin{figure}
	\begin{center}
		\tikzsetnextfilename{case_study2_time}
		\vspace{-0.15cm}
		\begin{tikzpicture}
\begin{axis}[
name=plot1,
  ylabel={Time / step [s]},
  legend pos=north west,
       width=0.95\columnwidth,
  height=3.2cm,
  ymin=0,
  ymax=0.9,
  xmin=0,
  xmax=1000,
    font={\sffamily},
  legend style={font=\footnotesize\sffamily},
  legend cell align={left},
  xticklabels={,,}]
\addplot[color=red,line width=1pt] table [x index=0,y index=1]{data/case_study2_time.dat};
\addplot[color=blue,line width=1pt] table [x index=0,y index=2]{data/case_study2_time.dat};
\legend{GP-SSM, GP-ASSM}
\end{axis}
\begin{axis}[
name=plot2,
    at=(plot1.below south west), anchor=above north west,
  xlabel={Time step $t$},
  ylabel={Total time [s]},
  legend pos=north west,
       width=0.95\columnwidth,
  height=3.2cm,
  ymin=0,
  ymax=400,
  xmin=0,
  xmax=1000,
    font={\sffamily},
  legend style={font=\footnotesize\sffamily},
  legend cell align={left}]
\addplot[color=red,line width=1pt] table [x index=0,y index=3]{data/case_study2_time.dat};
\addplot[color=blue,line width=1pt] table [x index=0,y index=4]{data/case_study2_time.dat};
\end{axis}
\end{tikzpicture} 
		\vspace{-0.5cm}\caption{Comparison of computational time between a standard GP-SSM (red) and the proposed GP-ASSM (blue). Top: The computational time per time step is constant for GP-ASSM in contrast to a cubic increase using GP-SSMs. Bottom: Total computational time of a sample scales linear for GP-ASSMs.}\vspace{-0.5cm}
		\label{fig:case_study2_time}
	\end{center}
\end{figure}
\subsection{Discussion}
In the previous sections, we show that the sampling of GPDMs, avoiding the impossible sampling of infinite dimensional objects, leads to non-Markovian dynamics. This characteristic is surprising as the representation of the GP-SSM and GP-NOE model, given by~\cref{for:gp_ssm,for:gp_noe} respectively, is based on a Markovian state space structure. However, the covariance term of the GP introduces dependencies across the states that leads to dependencies across time for GPDMs. Thus, the sampling of GP-SSMs and GP-NOE models generates non-Markovian dynamics, which we analyze from a control theoretical point of view. More precisely, a general description for approximated GPDMs based on a finite number of included past states/outputs is presented and compared against the true sampling. The approximation error of these models is analyzed with respect to the Kullback-Leibler divergence, the mean square prediction error and the variance of the prediction. Furthermore, we prove that the true variance of the next state ahead is always less than the variance of the approximated model as illustrated in~\cref{fig:case_study_1}. This is relevant for the usage of the approximation in variance based control approaches such as risk-sensitive control approaches, e.g.,~\cite{medina:icra2013b,likar2007predictive}. Additionally, the boundedness of GPDMs with bounded mean and variance functions, such as the commonly used squared exponential function, is proven and visualized in~\cref{fig:case_study_2}. The boundedness is an important property for the identification of unknown systems with GPDMs and is likewise exploited for robustness analysis in GPDM based control approaches, see~\cite{freeman1998robustness}. The introduced characteristics about the relation between the boundedness of the true sampling and the approximations allows a safe usage of the approximation. Finally, the approximated models allow not only a significant reduction of the total computing time as shown is~\cref{fig:case_study2_time} but also a constant computing time per step which enables the usages in real-time environments.
\section*{Conclusion}
In this article, we show that the sampling procedure for Gaussian process dynamical models leads to non-Markovian dynamics. We present a holistic description for approximated models which fulfills the Markov condition. The approximation error of these models is analyzed in respect to the Kullback-Leibler divergence, the mean square prediction error and the variance of the prediction. Furthermore, the boundedness of Gaussian process state space models and nonlinear output error models is qualitatively and quantitatively proven. We proof that the non-Markovian as well as the Markovian approximation is always bounded under specific conditions. Finally, we show the relation between different approximations with respect to the boundedness property of the system. Examples visualize the outcome and highlight the relevance of the results for data-driven based control approaches.
\section*{Acknowledgment}
This work was supported by the European Research Council (ERC) Consolidator Grant “Safe data-driven control for human-centric systems (COMAN)” agreement \#864686.

\ifCLASSOPTIONcaptionsoff
  \newpage
\fi

\bibliographystyle{IEEEtran}
\bibliography{mybib} 

\begin{thebibliography}{10}
\providecommand{\url}[1]{#1}
\csname url@samestyle\endcsname
\providecommand{\newblock}{\relax}
\providecommand{\bibinfo}[2]{#2}
\providecommand{\BIBentrySTDinterwordspacing}{\spaceskip=0pt\relax}
\providecommand{\BIBentryALTinterwordstretchfactor}{4}
\providecommand{\BIBentryALTinterwordspacing}{\spaceskip=\fontdimen2\font plus
\BIBentryALTinterwordstretchfactor\fontdimen3\font minus
  \fontdimen4\font\relax}
\providecommand{\BIBforeignlanguage}[2]{{%
\expandafter\ifx\csname l@#1\endcsname\relax
\typeout{** WARNING: IEEEtran.bst: No hyphenation pattern has been}%
\typeout{** loaded for the language `#1'. Using the pattern for}%
\typeout{** the default language instead.}%
\else
\language=\csname l@#1\endcsname
\fi
#2}}
\providecommand{\BIBdecl}{\relax}
\BIBdecl

\bibitem{nelles2013nonlinear}
O.~Nelles, \emph{Nonlinear system identification: from classical approaches to
  neural networks and fuzzy models}.\hskip 1em plus 0.5em minus 0.4em\relax
  Springer Science \& Business Media, 2013.

\bibitem{wang2008gaussian}
J.~M. Wang, D.~J. Fleet, and A.~Hertzmann, ``Gaussian process dynamical models
  for human motion,'' \emph{Transactions on Pattern Analysis and Machine
  Intelligence}, vol.~30, no.~2, pp. 283--298, 2008.

\bibitem{sigal2012loose}
L.~Sigal, M.~Isard, H.~Haussecker, and M.~J. Black, ``Loose-limbed people:
  Estimating 3d human pose and motion using non-parametric belief
  propagation,'' \emph{International journal of computer vision}, vol.~98,
  no.~1, pp. 15--48, 2012.

\bibitem{petelin2013evolving}
D.~Petelin, A.~Grancharova, and J.~Kocijan, ``Evolving {G}aussian process
  models for prediction of ozone concentration in the air,'' \emph{Simulation
  modelling practice and theory}, vol.~33, pp. 68--80, 2013.

\bibitem{hassouneh2012non}
I.~Hassouneh, T.~Serra, B.~K. Goodwin, and J.~M. Gil, ``Non-parametric and
  parametric modeling of biodiesel, sunflower oil, and crude oil price
  relationships,'' \emph{Energy Economics}, vol.~34, no.~5, pp. 1507--1513,
  2012.

\bibitem{soize2005comprehensive}
C.~Soize, ``A comprehensive overview of a non-parametric probabilistic approach
  of model uncertainties for predictive models in structural dynamics,''
  \emph{Journal of sound and vibration}, vol. 288, pp. 623--652, 2005.

\bibitem{lydia2014comprehensive}
M.~Lydia, S.~S. Kumar, A.~I. Selvakumar, and G.~E.~P. Kumar, ``A comprehensive
  review on wind turbine power curve modeling techniques,'' \emph{Renewable and
  Sustainable Energy Reviews}, vol.~30, pp. 452--460, 2014.

\bibitem{rasmussen2006gaussian}
C.~E. Rasmussen, \emph{{G}aussian processes for machine learning}.\hskip 1em
  plus 0.5em minus 0.4em\relax Citeseer, 2006.

\bibitem{frigola2014variational}
R.~Frigola, Y.~Chen, and C.~E. Rasmussen, ``Variational {G}aussian process
  state-space models,'' in \emph{Advances in Neural Information Processing
  Systems}, 2014, pp. 3680--3688.

\bibitem{kocijan2005dynamic}
J.~Kocijan, A.~Girard, B.~Banko, and R.~Murray-Smith, ``Dynamic systems
  identification with {G}aussian processes,'' \emph{Mathematical and Computer
  Modelling of Dynamical Systems}, vol.~11, no.~4, pp. 411--424, 2005.

\bibitem{ackermann2011nonlinear}
E.~R. Ackermann, J.~P. De~Villiers, and P.~Cilliers, ``Nonlinear dynamic
  systems modeling using {G}aussian processes: Predicting ionospheric total
  electron content over south africa,'' \emph{Journal of Geophysical Research:
  Space Physics}, vol. 116, no. A10, 2011.

\bibitem{kocijan2011output}
J.~Kocijan and D.~Petelin, ``Output-error model training for {G}aussian process
  models,'' in \emph{Proceedings of the Conference on Adaptive and Natural
  Computing Algorithms}.\hskip 1em plus 0.5em minus 0.4em\relax Springer, 2011,
  pp. 312--321.

\bibitem{rogers2011adaptive}
A.~Rogers, S.~Maleki, S.~Ghosh, and J.~Nicholas~R, ``Adaptive home heating
  control through {G}aussian process prediction and mathematical programming,''
  in \emph{Second International Workshop on Agent Technology for Energy Systems
  (ATES 2011)}, May 2011, pp. 71--78.

\bibitem{kocijan2004gaussian}
J.~Kocijan, R.~Murray-Smith, C.~E. Rasmussen, and A.~Girard, ``Gaussian process
  model based predictive control,'' in \emph{Proceedings of the American
  control conference (ACC)}, vol.~3.\hskip 1em plus 0.5em minus 0.4em\relax
  IEEE, 2004, pp. 2214--2219.

\bibitem{hewing2018cautious}
L.~Hewing, A.~Liniger, and M.~N. Zeilinger, ``Cautious {NMPC} with {G}aussian
  process dynamics for autonomous miniature race cars,'' in \emph{Proceedings
  of the European Control Conference (ECC)}.\hskip 1em plus 0.5em minus
  0.4em\relax IEEE, 2018, pp. 1341--1348.

\bibitem{kocijan2016modelling}
J.~Kocijan, \emph{Modelling and Control of Dynamic Systems Using {G}aussian
  Process Models}.\hskip 1em plus 0.5em minus 0.4em\relax Springer, 2016.

\bibitem{avzman2008non}
K.~A{\v{z}}man and J.~Kocijan, ``Non-linear model predictive control for models
  with local information and uncertainties,'' \emph{Transactions of the
  Institute of Measurement and Control}, vol.~30, no.~5, pp. 371--396, 2008.

\bibitem{wang2005gaussian}
J.~Wang, A.~Hertzmann, and D.~M. Blei, ``{G}aussian process dynamical models,''
  in \emph{Advances in neural information processing systems}, 2005, pp.
  1441--1448.

\bibitem{chowdhary2013bayesian}
G.~Chowdhary, H.~A. Kingravi, J.~P. How, and P.~A. Vela, ``Bayesian
  nonparametric adaptive control of time-varying systems using {G}aussian
  processes,'' in \emph{Proceedings of the American Control Conference
  (ACC)}.\hskip 1em plus 0.5em minus 0.4em\relax IEEE, 2013, pp. 2655--2661.

\bibitem{beckers:ecc2016}
T.~Beckers and S.~Hirche, ``Stability of {G}aussian process state space
  models,'' in \emph{Proceedings of the European Control Conference (ECC)},
  2016, pp. 2275--2281.

\bibitem{medina2015synthesizing}
J.~R. Medina, T.~Lorenz, and S.~Hirche, ``Synthesizing anticipatory haptic
  assistance considering human behavior uncertainty,'' \emph{Transactions on
  Robotics}, vol.~31, no.~1, pp. 180--190, 2015.

\bibitem{beckers2019automatica}
T.~Beckers, D.~Kulić, and S.~Hirche, ``Stable {G}aussian process based
  tracking control of {E}uler-{L}agrange systems,'' \emph{Automatica}, vol.
  103, pp. 390--397, 2019.

\bibitem{beckers:cdc2016}
T.~Beckers and S.~Hirche, ``Equilibrium distributions and stability analysis of
  {G}aussian process state space models,'' in \emph{Proceedings of the
  Conference on Decision and Control (CDC)}, 2016, pp. 6355--6361.

\bibitem{frigola2013bayesian}
R.~Frigola, F.~Lindsten, T.~B. Sch{\"o}n, and C.~E. Rasmussen, ``Bayesian
  inference and learning in {G}aussian process state-space models with particle
  mcmc,'' in \emph{Advances in Neural Information Processing Systems}, 2013,
  pp. 3156--3164.

\bibitem{wilson2020efficiently}
J.~Wilson, V.~Borovitskiy, A.~Terenin, P.~Mostowsky, and M.~Deisenroth,
  ``Efficiently sampling functions from {G}aussian process posteriors,'' in
  \emph{International Conference on Machine Learning}.\hskip 1em plus 0.5em
  minus 0.4em\relax PMLR, 2020, pp. 10\,292--10\,302.

\bibitem{cheng2017variational}
C.-A. Cheng and B.~Boots, ``Variational inference for {G}aussian process models
  with linear complexity,'' in \emph{Proceedings of the 31st International
  Conference on Neural Information Processing Systems}, 2017, p. 5190–5200.

\bibitem{bishop2006pattern}
C.~M. Bishop \emph{et~al.}, \emph{Pattern recognition and machine
  learning}.\hskip 1em plus 0.5em minus 0.4em\relax Springer New York, 2006,
  vol.~4.

\bibitem{mohammadi2016analytic}
H.~Mohammadi, R.~L. Riche, N.~Durrande, E.~Touboul, and X.~Bay, ``An analytic
  comparison of regularization methods for {G}aussian processes,'' \emph{arXiv
  preprint arXiv:1602.00853}, 2016.

\bibitem{rakitsch2013all}
B.~Rakitsch, C.~Lippert, K.~Borgwardt, and O.~Stegle, ``It is all in the noise:
  Efficient multi-task {G}aussian process inference with structured
  residuals,'' in \emph{Advances in neural information processing systems},
  2013, pp. 1466--1474.

\bibitem{melkumyan2011multi}
A.~Melkumyan and F.~Ramos, ``Multi-kernel {G}aussian processes,'' in
  \emph{Proceedings of the Twenty-Second International Joint Conference on
  Artificial Intelligence}, 2011, p. 1408–1413.

\bibitem{sjoberg1995nonlinear}
J.~Sj{\"o}berg, Q.~Zhang, L.~Ljung, A.~Benveniste, B.~Delyon, P.-Y. Glorennec,
  H.~Hjalmarsson, and A.~Juditsky, ``Nonlinear black-box modeling in system
  identification: a unified overview,'' \emph{Automatica}, vol.~31, no.~12, pp.
  1691--1724, 1995.

\bibitem{keviczky1999nonlinear}
R.~H.~L. Keviczky, \emph{Nonlinear system identification: input-output modeling
  approach}.\hskip 1em plus 0.5em minus 0.4em\relax Norwell, MA: Kluwer, 1999.

\bibitem{phan1970relationship}
M.~Phan and R.~Longman, ``Relationship between state-space and input-output
  models via observer markov parameters,'' \emph{WIT Transactions on The Built
  Environment}, vol.~22, 1970.

\bibitem{eleftheriadis2017identification}
S.~Eleftheriadis, T.~Nicholson, M.~Deisenroth, and J.~Hensman, ``Identification
  of {G}aussian process state space models,'' in \emph{Advances in neural
  information processing systems}, 2017, pp. 5309--5319.

\bibitem{umlauft:ecc2018}
J.~Umlauft, T.~Beckers, and S.~Hirche, ``A scenario-based optimal control
  approach for {G}aussian process state space models,'' in \emph{Proceedings of
  the European Control Conference (ECC)}, 2018, pp. 1386--1392.

\bibitem{matthews2016sparse}
A.~G. d.~G. Matthews, J.~Hensman, R.~Turner, and Z.~Ghahramani, ``On sparse
  variational methods and the {K}ullback-{L}eibler divergence between
  stochastic processes,'' \emph{Journal of Machine Learning Research}, vol.~51,
  pp. 231--239, 2016.

\bibitem{frigola2016bayesian}
R.~Frigola-Alcalde, ``Bayesian time series learning with {G}aussian
  processes,'' Ph.D. dissertation, University of Cambridge, 2016.

\bibitem{quinonero2005unifying}
J.~Qui{\~n}onero-Candela and C.~E. Rasmussen, ``A unifying view of sparse
  approximate {G}aussian process regression,'' \emph{Journal of Machine
  Learning Research}, vol.~6, no. Dec, pp. 1939--1959, 2005.

\bibitem{snelson2006sparse}
E.~Snelson and Z.~Ghahramani, ``Sparse {G}aussian processes using
  pseudo-inputs,'' in \emph{Advances in neural information processing systems},
  2006, pp. 1257--1264.

\bibitem{liu2010note}
X.~Liu, ``A note on the existence of periodic solutions in discrete
  predator--prey models,'' \emph{Applied Mathematical Modelling}, vol.~34,
  no.~9, pp. 2477--2483, 2010.

\bibitem{sprott2007labyrinth}
J.~C. Sprott and K.~E. Chlouverakis, ``Labyrinth chaos,'' \emph{International
  Journal of Bifurcation and Chaos}, vol.~17, no.~06, pp. 2097--2108, 2007.

\bibitem{medina:icra2013b}
J.~M. Hernández, D.~Sieber, and S.~Hirche, ``Risk-sensitive interaction
  control in uncertain manipulation tasks,'' in \emph{Proceedings of the
  International Conference on Robotics and Automation}, 2013, pp. 502--507.

\bibitem{likar2007predictive}
B.~Likar and J.~Kocijan, ``Predictive control of a gas--liquid separation plant
  based on a {G}aussian process model,'' \emph{Computers \& chemical
  engineering}, vol.~31, no.~3, pp. 142--152, 2007.

\bibitem{freeman1998robustness}
R.~A. Freeman, M.~Krsti{\'c}, and P.~Kokotovi{\'c}, ``Robustness of adaptive
  nonlinear control to bounded uncertainties,'' \emph{Automatica}, vol.~34,
  no.~10, pp. 1227--1230, 1998.

\end{thebibliography}

%
\vspace{-0.4cm}
\begin{IEEEbiography}[{\includegraphics[width=1in,height=1.25in,clip,keepaspectratio]{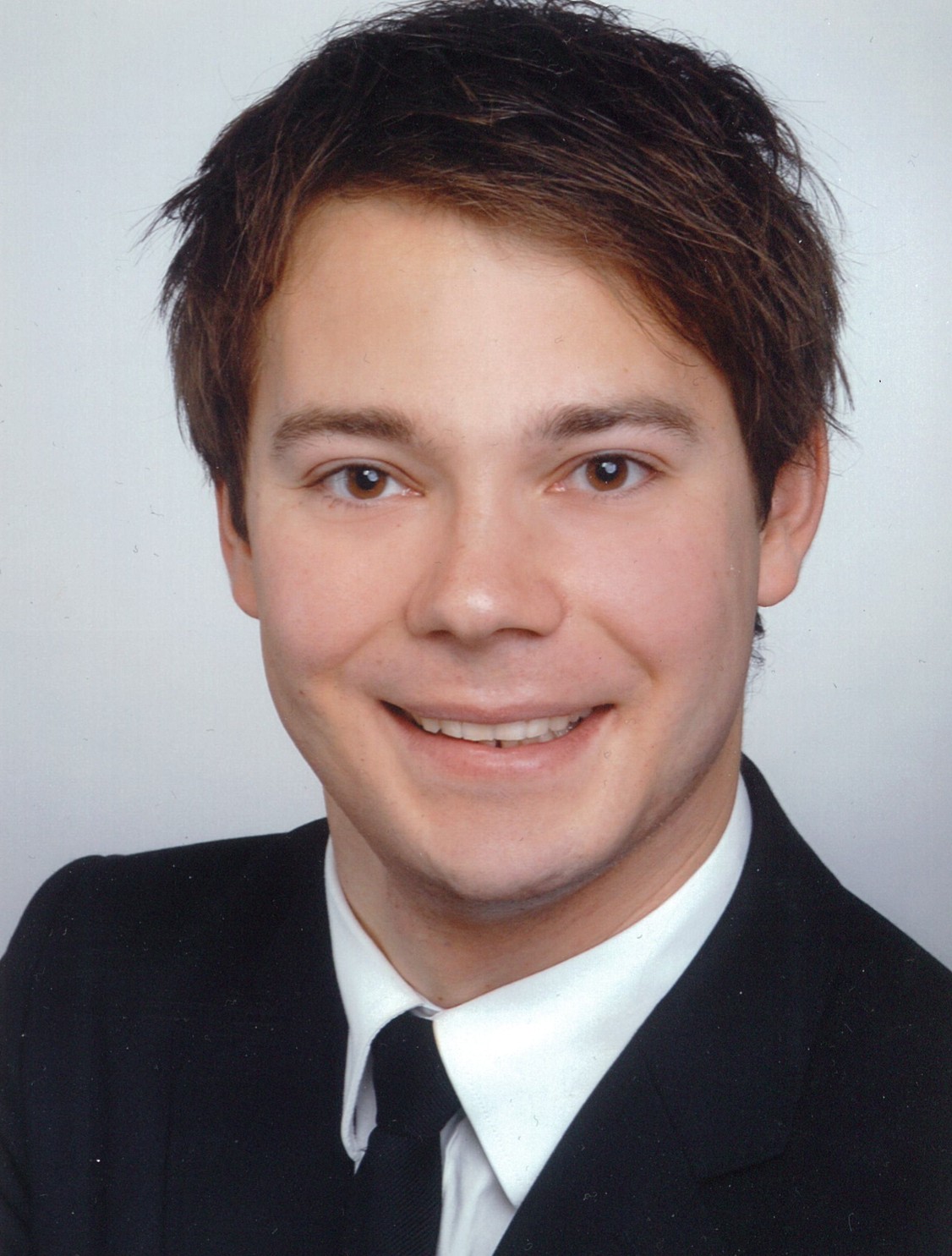}}]{Thomas Beckers}
Thomas Beckers is a postdoctoral researcher at the Department of Electrical and Systems Engineering, University of Pennsylvania. In 2020, he successfully defended his PhD thesis at the Technical University of Munich (TUM), Germany. He received the B.Sc. and M.Sc. degree in electrical engineering in 2010 and 2013, respectively, from the Technical University of Braunschweig, Germany. In 2018, he was a visiting researcher at the University of California, Berkeley. His research interests include data-driven based identification and control, nonparametric systems, and formal methods for safe learning.
\end{IEEEbiography}
\vspace{-0.4cm}
\begin{IEEEbiography}[{\includegraphics[width=1in,height=1.25in,clip,keepaspectratio]{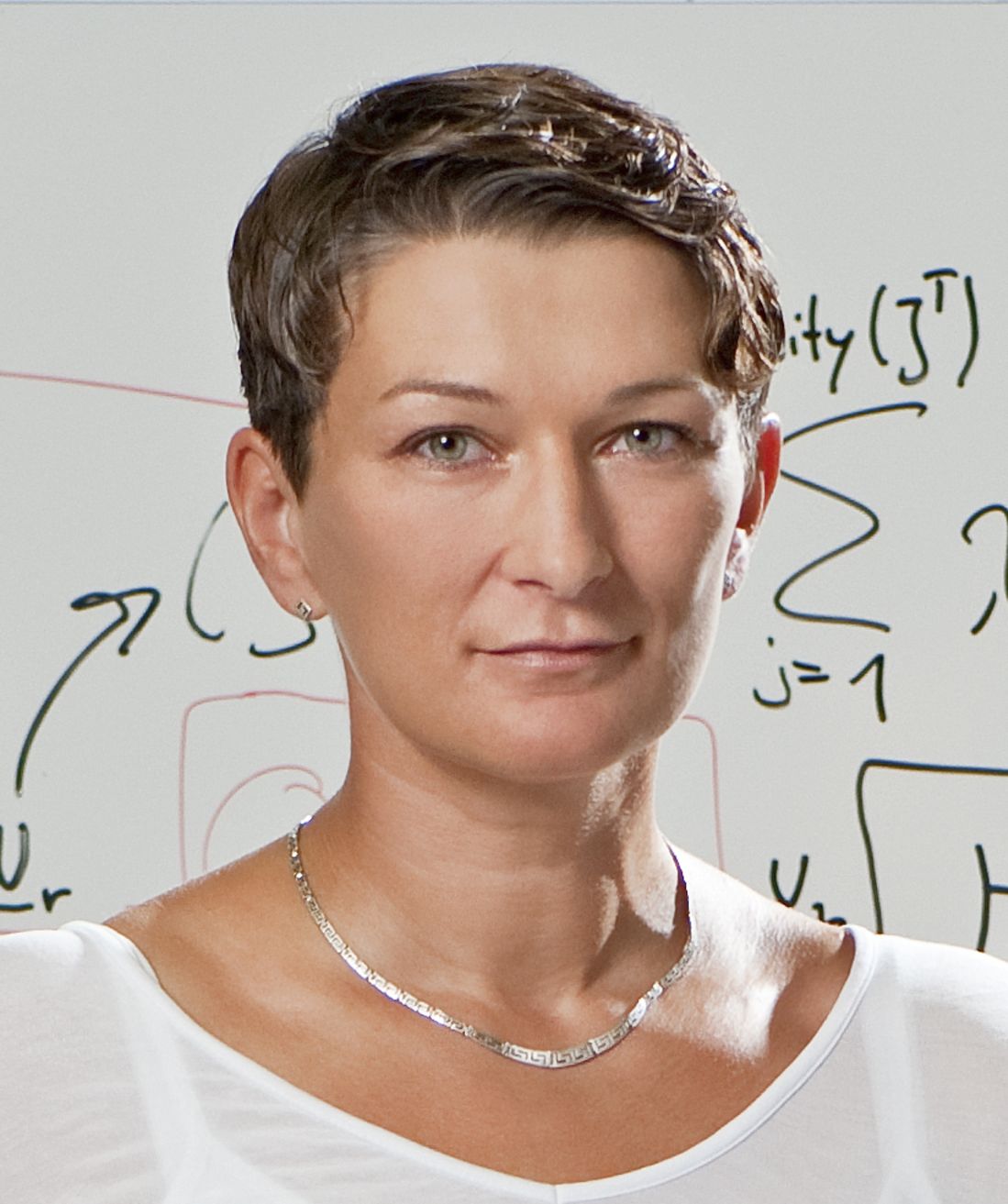}}]{Sandra Hirche}
Sandra Hirche received the Diplom-Ingenieur degree in aeronautical engineering from Technical University Berlin, Germany, in 2002 and the Doktor-Ingenieur degree in electrical engineering from Technical University Munich, Germany, in 2005. From 2005 to 2007 she was awarded a Postdoc scholarship from the Japanese Society for the Promotion of Science at the Fujita Laboratory, Tokyo Institute of Technology, Tokyo, Japan. From 2008 to 2012 she has been an associate professor at Technical University Munich. Since 2013 she is TUM Liesel Beckmann Distinguished Professor and has the Chair of Information-oriented Control in the Department of Electrical and Computer Engineering at Technical University Munich. Her main research interests include cooperative and distributed networked control as well as learning control with applications in human-robot interaction, multi-robot systems, and general robotics. She has published more than 150 papers in international journals, books, and refereed conferences.
\end{IEEEbiography}







\end{document}